\documentclass[a4paper,USenglish,cleveref,autoref,thm-restate]{lipics-v2021}

\pdfoutput=1 
\hideLIPIcs

\makeatletter
\ifx\@hideLIPIcs\@undefined\else
\let\@oddfoot\@empty
\fi
\makeatother
\nolinenumbers

\title{Neighborhood-Aware Graph Labeling Problem}
\titlerunning{Neighborhood-Aware Graph Labeling Problem}

\author{Mohammad Shahverdikondori\footnote{Equal contribution.
Author order between the first two authors was determined randomly.}}{EPFL, Switzerland}{mohammad.shahverdikondori@epfl.ch}{}{}
\author{Sepehr Elahi\footnotemark[1]}{EPFL, Switzerland}{sepehr.elahi@epfl.ch}{}{}
\author{Patrick Thiran}{EPFL, Switzerland}{patrick.thiran@epfl.ch}{}{}
\author{Negar Kiyavash}{EPFL, Switzerland}{negar.kiyavash@epfl.ch}{}{}

\authorrunning{M.\ Shahverdikondori, S.\ Elahi, P.\ Thiran, N.\ Kiyavash}



\relatedversion{}

\EventLongTitle{53rd EATCS International Colloquium on Automata, Languages, and Programming (ICALP)}
\EventShortTitle{ICALP 2026}
\EventAcronym{ICALP}
\EventYear{2026}
\EventDate{July 7--10, 2026}
\EventLocation{London, United Kingdom}

\usepackage[utf8]{inputenc} 
\usepackage[T1]{fontenc}    
\usepackage{url}            
\usepackage{cite}           
\usepackage{hyperref}       
\usepackage{booktabs}       
\usepackage{amsfonts}       
\usepackage{nicefrac}       
\usepackage{microtype}      
\usepackage{xcolor}         

\usepackage{mathtools} 
\usepackage{booktabs} 
\usepackage{tikz} 

\usepackage{algorithm}
\usepackage{algorithmic}

\usepackage{microtype}
\usepackage{graphicx}
\usepackage{booktabs}

\usepackage{amsthm}
\usepackage{thmtools,thm-restate}

\usepackage{hyperref}
\usepackage{caption}
\usepackage{subcaption}
\usepackage{bbm}

\usepackage{amsmath}
\usepackage{amssymb}
\usepackage{mathtools}

\usepackage{mwe}
\usepackage{tikz} 
\usepackage{pgf}
\usepackage{bbm}
\usepackage{extarrows} 

\newcommand{\lcal}[0]{\mathcal{L}}

\usepackage{xspace}
\newcommand{\name}{\textsc{NAGL}\xspace}

\newcommand{\Ncl}[1]{N[#1]}

\newcommand{\sumv}[0]{\sum_{v \in V}}
\newcommand{\OPT}[0]{\mathrm{OPT}}

\DeclareMathOperator*{\argmax}{arg\,max}

\DeclareMathOperator{\tw}{tw}

\DeclareMathOperator{\dist}{dist}

\theoremstyle{plain}

\usepackage[textsize=tiny]{todonotes}

\begin{document}

\maketitle

\begin{abstract}
Motivated by optimization oracles in bandits with network interference, we study the Neighborhood-Aware Graph Labeling (\name) problem. Given a graph $G = (V,E)$, a label set of size $L$, and local reward functions $f_v$ accessed via evaluation oracles, the objective is to assign labels to maximize $\sum_{v \in V} f_v(x_{\Ncl{v}})$, where each term depends on the closed neighborhood of $v$.
Two vertices co-occur in some neighborhood term exactly when their distance in $G$ is at most $2$, so the dependency graph is the squared graph $G^2$ and $\tw(G^2)$ governs exact algorithms and matching fine-grained lower bounds.
Accordingly, we show that this dependence is inherent: \name is NP-hard even on star graphs with binary labels and, assuming SETH, admits no $(L-\varepsilon)^{\tw(G^2)}\cdot n^{O(1)}$-time algorithm for any $\varepsilon>0$.
We match this with an exact dynamic program on a tree decomposition of $G^2$ running in $O\!\left(n\cdot \tw(G^2)\cdot L^{\tw(G^2)+1}\right)$ time.
For approximation, unless $\mathsf{P}=\mathsf{NP}$, for every $\varepsilon>0$ there is no polynomial-time $n^{1-\varepsilon}$-approximation on general graphs even under the promise $\OPT>0$; 
without the promise $\OPT>0$, no finite multiplicative approximation ratio is possible.
In the nonnegative-reward regime, we give polynomial-time approximation algorithms for \name in two settings: (i) given a proper $q$-coloring of $G^2$, we obtain a $1/q$-approximation; and (ii) on planar graphs of bounded maximum degree, we develop a Baker-type polynomial-time approximation scheme (PTAS), which becomes an efficient PTAS (EPTAS) when $L$ is constant.
\end{abstract}

\section{Introduction}

A wide range of combinatorial optimization problems on graphs are characterized by the presence of \emph{local interactions} coupled with a \emph{global objective}.
In such problems, decisions made at individual vertices or edges affect only a small neighborhood, yet these local effects overlap and interact, making the global optimization problem highly nontrivial.
This phenomenon arises in numerous settings, including graph labeling problems, graphical models and constraint networks \cite{koller2009pgm,wainwright2008graphical,dechter1989treeclustering}, and graph-based energy minimization and labeling formulations \cite{boykov2001graphcuts,kolmogorov2004graphcuts}, as well as networked decision-making on graphs \cite{agarwal2024networkinterference, jamshidi2025graph}.
Understanding how local dependencies influence global computational complexity is a central theme in algorithmic graph theory.

In this work, we study the Neighborhood-Aware Graph Labeling (\name) problem.
The input is an undirected graph $G = (V, E)$ and a finite label set $\lcal$ with $L \coloneqq |\lcal|$.
Each vertex $v$ is assigned a label $x_v\in\lcal$ and receives a numerical reward that is a function of the labels in its closed neighborhood $\Ncl{v}$ (i.e., the neighborhood including the vertex itself).
The goal is to maximize the total reward $F(x)=\sum_{v\in V} f_v(x_{\Ncl{v}})$.
The computational difficulty of \name stems from the objective being a sum of \emph{overlapping} neighborhood terms: modifying a single label $x_u$ can change every local reward $f_v(x_{\Ncl{v}})$ for which $u\in \Ncl{v}$.
This overlap is captured by the squared graph $G^2$: two vertices can appear together in some local term $f_w$ exactly when their distance in $G$ is at most $2$.
In particular, each closed neighborhood $\Ncl{v}$ induces a clique in $G^2$, making $\tw(G^2)$ the appropriate treewidth parameter for exact dynamic programming.

The primary motivation for studying \name is to address the computational tractability of algorithms arising in recent work on multi-armed bandits with network interference \cite{interference1-jia2024multi, agarwal2024networkinterference, jamshidi2025graph}.
In this class of problems, each action corresponds to assigning labels (e.g., treatments or interventions) to all nodes of an interference graph, and the reward of each node depends not only on its own assigned label but also on the labels assigned to its neighbors.
Such models naturally arise, for example, in vaccination or information diffusion scenarios, where the treatment of an individual may affect the health outcomes or behavior of nearby individuals in a social or contact network \cite{elahi2025vacc}.
Learning algorithms for bandits with interference typically maintain estimates or upper confidence bounds on the local reward functions, and then assume access to an oracle that computes a globally optimal labeling given these local estimates.
See, e.g., \cite{jamshidi2025graph,agarwal2024networkinterference}.
This oracle problem is precisely an instance of \name. While the statistical aspects of learning under interference have received attention, the computational complexity of this underlying optimization problem has not been previously examined. Characterizing the hardness of this oracle and developing efficient algorithms are therefore crucial for the practical applicability of these learning methods.
Our results provide a tight characterization of the computational feasibility of this oracle.
More broadly, \name can be viewed as a graph labeling problem with higher-order (neighborhood) rewards, generalizing classical pairwise formulations used, e.g., in energy minimization via graph cuts \cite{boykov2001graphcuts,kolmogorov2004graphcuts} and inference in graphical models \cite{koller2009pgm,wainwright2008graphical}.

Equivalently, \name can be viewed as a weighted constrained satisfaction problem (WCSP) 
with primal graph $G^2$ \cite{dechter1989treeclustering,freuder1990ktree}, which provides intuition for why $\tw(G^2)$ governs the complexity of exact algorithms, via treewidth-based inference methods \cite{dechter1989treeclustering,lauritzen1988local}.
Our model, algorithms, and proofs are done in \name's graph-native language motivated by network interference optimization, giving self-contained arguments specialized to squared-neighborhood structure rather than using WCSP machinery.
This graph-first perspective is also what enables our approximation regimes both in terms of a coloring of the squared graph and a Baker-type PTAS/EPTAS for planar graphs, which exploit structure beyond $\tw(G^2)$ \cite{baker1994approximation}.

We further discuss the related work in \Cref{app:related_work},
deferred proofs appear in \Cref{apd:proofs}, and
experimental results are reported in \Cref{apd:experiments}.

\vspace{0.5\baselineskip}
\noindent
\textbf{Contributions.} Our main contributions are as follows:
\begin{itemize}
    \item We prove that \name is NP-hard under the standard evaluation-oracle model for local reward functions.
    We further show that, assuming SETH, there exists no algorithm running in time $(L-\varepsilon)^{\tw(G^2)}\cdot n^{O(1)}$ for any $\varepsilon>0$.
    We also establish strong inapproximability for the general version of \name (allowing negative rewards).
    In particular, unless $\mathsf{P}=\mathsf{NP}$, \name admits no polynomial-time algorithm with any bounded multiplicative approximation guarantee, even on star graphs.
    Moreover, unless $\mathsf{P}=\mathsf{NP}$, for every $\varepsilon>0$ there exists no polynomial-time $n^{1-\varepsilon}$-approximation algorithm for general graphs, even on instances with strictly positive optimum.
    \item We introduce \emph{Clique-Focused Dynamic Programming} (CFDP), an exact dynamic programming algorithm over a tree decomposition of $G^2$ with running time $O\bigl(n\cdot \tw(G^2)\cdot L^{\tw(G^2)+1}\bigr)$.
    For constant $L$ this yields a single-exponential fixed-parameter tractable (FPT) algorithm parameterized by $\tw(G^2)$.
    Our SETH-based lower bound proves that the base $L$ in the running time is optimal up to polynomial factors.
    \item In the nonnegative-reward regime, given a proper $q$-coloring of $G^2$ we obtain a $1/q$-approximation in time $O\!\left(L^{\Delta(G)+1}\right)\cdot n^{O(1)}$.
    Without a coloring as input, a greedy coloring yields a $1/(\Delta(G^2)+1)$-approximation, and hence a $1/(\Delta(G)^2+1)$-approximation using $\Delta(G^2)\le \Delta(G)^2$.
    \item In the nonnegative-reward regime, if $G$ is planar with $\Delta(G)$ a constant, we obtain a Baker-type PTAS for \name.
    For every $\varepsilon\in(0,1)$, we compute a $(1-\varepsilon)$-approximate labeling in time $L^{O(\Delta(G)^2/\varepsilon)}\cdot n^{O(1)}$.
    For constant $L$, this yields an EPTAS.

\end{itemize}

\section{Preliminaries and problem setup} \label{sec: setup}
We define \name and recall the graph-theoretic notions used throughout the paper.

We model interactions by an undirected graph $G=(V,E)$ with $n\coloneqq |V|$.
For $m\in\mathbb{N}$, we write $[m]\coloneqq \{1,2,\ldots,m\}$.
For each $v\in V$, let $\Ncl{v} \coloneqq \{ v \} \cup \{ u \in V : \{u,v\} \in E \}$ denote its closed neighborhood.
Let $\lcal$ be a finite set of labels and write $L \coloneqq |\lcal|$.
We view $L$ as an explicit parameter and keep track of its dependence in our running-time bounds.
Function $x: V \to \lcal$, called a labeling, specifies the labels assigned to all vertices in $G$.\footnote{Generalization of the results to finite label sets with vertex-dependent cardinalities is straightforward.}
For any subset $S \subseteq V$, we denote by $x_S$ the restriction of $x$ to the vertices in $S$.

For each vertex $v \in V$, the function $f_v : \lcal^{|\Ncl{v}|} \to \mathbb{R}$, called the reward function of vertex $v$, specifies the reward of this vertex as a function of the labels assigned to its closed neighborhood.
We assume an evaluation-oracle model in which each $f_v$ is given by a succinct representation supporting polynomial-time evaluation: given $v\in V$ and an assignment $a:\Ncl{v}\to\lcal$, we can compute $f_v(a)$ in time polynomial in the instance encoding size, and oracle values have polynomial-bit encoding, so the decision version of \name lies in NP.

While we allow rewards to be negative in general, our multiplicative approximation statements will be stated only under explicit nonnegativity assumptions (see \Cref{sec:color-approx} and \Cref{sec:planar-ptas}).

The objective of the \name problem is to find a labeling $x^\star \in \lcal^V$ that maximizes the total reward, that is,
\begin{equation} \label{eq:NAGL} 
    x^\star \in \argmax_{x \in \lcal^V} F(x), \qquad F(x) = \sumv f_v(x_{\Ncl{v}}). 
\end{equation}
We write $\OPT \coloneqq \max_{x\in\lcal^V} F(x)=F(x^\star)$ for the optimal value.

We study the computational complexity of solving \eqref{eq:NAGL} as a function of structural properties of the underlying graph $G$.
To keep the exposition self-contained, we briefly review the relevant graph-theoretic notions used throughout the paper.
\begin{definition}[Squared Graph] \label{def: squared graph}
Given an undirected graph $G = (V,E)$, the squared graph of $G$, denoted by $G^2$, is the graph on the same vertex set $V$ in which two distinct vertices $u,v \in V$ are connected by an edge if and only if their distance in $G$ is at most two.
\end{definition}

\begin{definition}[Tree Decomposition] \label{def: tree decomposition}
A tree decomposition of an undirected graph $G = (V,E)$ is a pair $(T, \{X_i\}_{i \in I})$, where $T$ is a tree and each node $i \in I$ is associated with a subset $X_i \subseteq V$, called a bag, such that the following conditions hold:
\begin{enumerate}
    \item For every vertex $v \in V$, there exists at least one bag $X_i$ such that $v \in X_i$.
    \item For every edge $\{u,v\} \in E$, there exists a bag $X_i$ such that $\{u,v\} \subseteq X_i$.
    \item For every vertex $v \in V$, the set of nodes $\{ i \in I : v \in X_i \}$ induces a connected subtree of $T$.
\end{enumerate}
\end{definition}

\begin{definition}[Treewidth] \label{def: treewidth}
The width of a tree decomposition $(T, \{X_i\}_{i \in I})$ is defined as $\max_{i \in I} |X_i| - 1$. The treewidth of a graph $G$, denoted by $\tw(G)$, is the minimum width over all tree decompositions of $G$.
\end{definition}

Treewidth is a graph parameter that quantifies how ``tree-like'' a graph is; for instance, every tree has treewidth $1$.
A large class of NP-hard problems becomes tractable on graphs of bounded treewidth via dynamic programming over a tree decomposition, with run time typically exponential in $\tw(G)$ and polynomial in $n$; prominent examples include \textsc{Maximum Independent Set} and \textsc{Max-Cut} \cite{cygan2015parameterized}.
Closely related notions also play a central role in probabilistic graphical models, where the complexity of exact inference scales exponentially with the treewidth \cite{koller2009pgm,wainwright2008graphical}.

For \name, although the objective \eqref{eq:NAGL} is defined over $G$, the relevant dependency structure is captured by the squared graph $G^2$.
To see the intuition behind this, observe that each term $f_v$ depends on all the labels in $\Ncl{v}$; consequently, any two distinct \textit{interacting} vertices are adjacent in $G^2$.
Motivated by this observation, we will express both algorithmic and hardness statements in terms of the label-set size $L$ and the parameter $\tw(G^2)$, which may be substantially larger than $\tw(G)$.

The following remark shows that the model extends naturally to settings in which each reward function depends on labels beyond the immediate neighborhood.

\begin{remark}
Suppose the reward function $f_v$ of a vertex $v$ depends on the labels of a set $D(v) \subseteq V$ (not necessarily equal to $\Ncl{v}$).
This can be captured as an instance of \name on an augmented graph $H$ on vertex set $V$ in which $v$ is adjacent to every vertex in $D(v)\setminus\{v\}$.
All statements in the paper apply to $H$ by replacing $G$ with $H$.
In particular, if each $f_v$ depends on all vertices within distance at most $k$ from $v$ in $G$,  $H = G^k$.
\end{remark}

\section{Computational hardness of \name} \label{sec:hardness}

In this section, we establish computational hardness results for \name under the evaluation-oracle input model described in \Cref{sec: setup}.
We give a reduction from $k$-\textsc{SAT} (and hence already from \textsc{3-SAT}) showing that \name is NP-hard even on star graphs.
This reduction yields a SETH-tight lower bound in terms of $\tw(G^2)$ and $L$ that we will later match with a dynamic programming algorithm in \Cref{sec:alg}.
We then complement these exact lower bounds with strong inapproximability.

\subsection{SETH-tight lower bounds via SAT on a star graph}
\label{subsec:hardness-sat}

\begin{definition}[Strong Exponential Time Hypothesis (SETH) \cite{Impagliazzo2001kSAT, lokshtanovMarxSaurabhSODA11}]
SETH asserts that for every $\delta>0$ there exists an integer $k\ge 3$ such that $k$-\textsc{SAT} on $N$ variables cannot be solved in time $(2-\delta)^N\cdot N^{O(1)}$.
\end{definition}

Assuming SETH, we establish a tight lower bound for \name parameterized by $\tw(G^2)$ and $L$ via a reduction from \textsc{$k$-SAT} to \name on a star graph, as stated in the following lemma.

\begin{lemma}[$k$-SAT to \name on a star]
\label{lem:sat-star}
Fix an integer $L\ge 2$.
Given a $k$-CNF formula $\varphi$ with $N$ variables, one can construct in polynomial time an instance of \name with label set $\lcal=[L]$ on a star graph $G$ with $t+1$ vertices, where $t\coloneqq \lceil N/\log_2 L\rceil$, such that all local rewards take values in $\{0,1\}$ and
the optimal value satisfies $\OPT=1$ if $\varphi$ is satisfiable and $\OPT=0$ otherwise.
Moreover, $G^2$ is a clique on $t+1$ vertices and hence $\tw(G^2)=t$.
\end{lemma}

\begin{proof}
Let $t\coloneqq \lceil N/\log_2 L\rceil$, so that $L^t \ge 2^N$.
Let $G$ be the star with center $c$ and leaves $u_1,\ldots,u_t$.
For each leaf $u_i$, set $f_{u_i}\equiv 0$.
We define the center reward $f_c$ to ignore the label of $c$ and depend only on the leaf labels.

Fix a labeling $x\in \lcal^V$.
For each $i\in[t]$, map $x_{u_i}\in\{1,\ldots,L\}$ to the digit $d_i\coloneqq x_{u_i}-1\in\{0,\ldots,L-1\}$ and interpret $(d_1,\ldots,d_t)$ as the base-$L$ integer $I(x)\coloneqq \sum_{i=1}^t d_i\,L^{i-1}$.
If $I(x)\ge 2^N$, set $f_c(x_{\Ncl{c}})\coloneqq 0$.
Otherwise, let $\sigma(x)\in\{0,1\}^N$ be the $N$-bit binary representation of $I(x)$ (padded with leading zeros), and define
$f_c(x_{\Ncl{c}})\coloneqq 1$ iff $\varphi$ evaluates to true under the assignment $\sigma(x)$, and $0$ otherwise.
Crucially, this evaluation takes time polynomial in $|\varphi|$ (and in $t\log L$ to compute $I(x)$), hence polynomial in the instance size.

Since all leaves have reward $0$, we have $F(x)=f_c(x_{\Ncl{c}})\in\{0,1\}$ for every labeling $x$.
If $\varphi$ is satisfiable, let $\sigma^\star\in\{0,1\}^N$ be a satisfying assignment and let $I^\star<2^N$ be its integer value.
Because $L^t\ge 2^N$, $I^\star$ has a base-$L$ representation using at most $t$ digits, which yields labels for the leaves and hence a labeling $x$ with $I(x)=I^\star$.
For this labeling, $\sigma(x)=\sigma^\star$ and thus $F(x)=1$.
Conversely, if $\varphi$ is unsatisfiable then for every labeling $x$ we have $f_c(x_{\Ncl{c}})=0$, and hence $F(x)=0$.
Therefore $\OPT=1$ if and only if $\varphi$ is satisfiable.

Finally, any two distinct leaves are at distance $2$ in $G$, so $G^2$ is the clique $K_{t+1}$.
Hence $\tw(G^2)=\tw(K_{t+1})=t$.
\end{proof}

\begin{remark}[Why the oracle model matters]
Lemma~\ref{lem:sat-star} relies on the fact that each local reward function is accessed via a succinct representation that supports polynomial-time evaluation.
If $f_c$ were given explicitly as a table over all assignments to $\Ncl{c}$, then the input would have size $\Theta(L^{|\Ncl{c}|})=\Theta(L^{t+1})$, which is exponential in $N$ (since $L^t \ge 2^N$).
Thus, the reduction is polynomial-time only in the evaluation-oracle model.
\end{remark}

\begin{theorem}[NP-hardness on star graphs]
\label{thm:np-hard-star}
The \name problem is NP-hard even when the input graph $G$ is a star and the label set is binary.
\end{theorem}

\begin{proof}
We reduce from $3$-\textsc{SAT}, which is NP-hard.
Given an instance $\varphi$ with $N$ variables, apply \Cref{lem:sat-star} with $L=2$ to construct in polynomial time an \name instance on a star graph with binary labels.
By \Cref{lem:sat-star}, the constructed instance satisfies $\OPT=1$ if and only if $\varphi$ is satisfiable.
\end{proof}

\begin{theorem}[SETH-tight lower bound parameterized by $\tw(G^2)$]
\label{thm:seTH-tight}
Fix an integer $L\ge 2$.
Assuming SETH, for every $\varepsilon>0$ there is no algorithm that solves \name with label set size $L$ in time $(L-\varepsilon)^{\tw(G^2)}\cdot n^{O(1)}$.
This holds even when the input graph is restricted to be a star and all local rewards take values in $\{0,1\}$.
\end{theorem}

\begin{proof}
Fix $L\ge 2$ and $\varepsilon>0$.
If $L-\varepsilon \le 1$ (equivalently, $\varepsilon \ge L-1$), then the claimed running time bound is polynomial in $n$; combined with \Cref{lem:sat-star} this would contradict SETH.
Hence we may assume $1 < L-\varepsilon < L$.

Assume, for contradiction, that there is an algorithm $\mathcal{A}$ that solves \name in time
$(L-\varepsilon)^{\tw(G^2)}\cdot n^{O(1)}$.
Define $\gamma\coloneqq \frac{\log_2(L-\varepsilon)}{\log_2 L}$,
so $0<\gamma<1$.
Let $(2-\delta)\coloneqq 2^{\gamma}$, i.e., $\delta\coloneqq 2-2^{\gamma}>0$.
By SETH, there exists an integer $k\ge 3$ such that $k$-\textsc{SAT} on $N$ variables cannot be solved in time $(2-\delta)^N\cdot N^{O(1)}$.

Let $k$ be this integer, and let $\varphi$ be an instance of $k$-\textsc{SAT} with $N$ variables.
Construct the corresponding \name instance via \Cref{lem:sat-star}.
The constructed graph is a star and satisfies $\tw(G^2)=t$ where $t=\lceil N/\log_2 L\rceil$.
Moreover, the instance has $n=t+1=O(N)$ vertices, so $n^{O(1)} \le N^{O(1)}$.
Running $\mathcal{A}$ on this instance decides satisfiability of $\varphi$ in time
\begin{align*}
    (L-\varepsilon)^t\cdot n^{O(1)}
    &\le
    (L-\varepsilon)^{N/\log_2 L + 1}\cdot N^{O(1)} \\
    &=
    (L-\varepsilon)\cdot 2^{\gamma N}\cdot N^{O(1)} \\
    &=
    2^{\gamma N}\cdot N^{O(1)}
    \qquad\text{(absorbing constants into $N^{O(1)}$)}
    =
    (2-\delta)^N\cdot N^{O(1)}.
\end{align*}
This contradicts SETH.
\end{proof}

\subsection{Inapproximability} \label{subsec:inapprox}
We next establish hardness of approximation for \name.
We use the standard multiplicative approximation ratio for maximization problems with nonnegative objective values (all hardness instances in this section have rewards in $\{0,1\}$).
Given a function $\rho:\mathbb{N}\to[1,\infty)$, a polynomial-time $\rho(n)$-approximation algorithm is required to output, on every instance on $n$ vertices, a labeling $x$ such that $F(x)\ge \OPT/\rho(n)$.

\begin{lemma}[Gap-at-zero barrier for multiplicative approximation]
\label{lem:gap-zero}
Unless $\mathsf{P}=\mathsf{NP}$, for every function $\rho:\mathbb{N}\to[1,\infty)$ there is no polynomial-time $\rho(n)$-approximation algorithm for \name, even when $G$ is a star and all local rewards take values in $\{0,1\}$.
\end{lemma}

\begin{proof}
Using \Cref{lem:sat-star}, a multiplicative approximation would distinguish $\OPT=0$ from $\OPT=1$ and decide satisfiability; see \Cref{apd:inapprox-proofs}.
\end{proof}

The previous lemma is a gap-at-zero barrier: without a promise that $\OPT>0$, no algorithm can guarantee any finite multiplicative ratio.
To obtain an inapproximability statement that remains informative under the promise $\OPT>0$, we next reduce from the \textsc{Maximum Independent Set} (\textsc{MIS}) problem.
Given a graph $G=(V,E)$, a set $I\subseteq V$ is independent if it contains no edge of $G$, and we write $\alpha(G)$ for the maximum size of an independent set.
The reduction below produces \name instances with $\OPT=\alpha(G)$, and hence $\OPT\ge 1$ whenever $|V|\ge 1$.

\begin{theorem}[$n^{1-\varepsilon}$-inapproximability of \name]
\label{thm:inapprox}
Unless $\mathsf{P}=\mathsf{NP}$, for every $\varepsilon>0$ there is no polynomial-time $n^{1-\varepsilon}$-approximation algorithm for \name.
This holds even on instances with $L=2$ and $\OPT > 0$.
\end{theorem}

\begin{proof}
We reduce from \textsc{MIS} by encoding independent sets as labelings with $F(x)=|I(x)|$; see \Cref{apd:inapprox-proofs}.
\end{proof}

\subsection{On the parameter $\tw(G^2)$}
\label{subsec:twG2}

Our SETH-based lower bound rules out algorithms running in time $(L-\varepsilon)^{\tw(G^2)}\cdot n^{O(1)}$ (for constant $L$), and in \Cref{sec:alg} we give a matching dynamic program.
We briefly relate $\tw(G^2)$ to more standard structural parameters of $G$.

A first takeaway is that bounding $\tw(G)$ alone does not control $\tw(G^2)$.
For instance, if $G$ is a star on $n$ vertices, then $\tw(G)=1$, but $G^2$ is the clique $K_n$ and thus $\tw(G^2)=n-1$.
This is consistent with \Cref{thm:np-hard-star}, which shows that \name remains NP-hard even on star graphs.
Bounding the maximum degree $\Delta(G)$ alone is also insufficient: there are bounded-degree graph families with $\tw(G)=\Omega(n)$ (e.g., bounded-degree expanders), and since $G$ is a subgraph of $G^2$ we have $\tw(G^2)\ge \tw(G)=\Omega(n)$ \cite{tw-small-delta-grohe2009tree, tw-small-delta2-feige2016giant}.

On the positive side, $\tw(G^2)$ can be bounded in terms of $\tw(G)$ and $\Delta(G)$.
A standard bag-expansion argument gives $\tw(G^2)\le (\tw(G)+1)\bigl(\Delta(G^2)+1\bigr)-1 \le (\tw(G)+1)\bigl(1+\Delta(G)^2\bigr)-1$, where we used $\Delta(G^2)\le \Delta(G)^2$; see, e.g., the discussion on the treewidth of the graph powers in \cite{power-treewidth-gurski2025behavior}.


\section{Algorithm and upper bound}
\label{sec:alg}

We give an exact algorithm for \name by dynamic programming over a tree decomposition of the squared graph $G^2$.
The algorithm assumes that a width-$t$ tree decomposition of $G^2$ is given as part of the input; computing such a decomposition can be done using standard algorithms with an additional $f(t)\,n^{O(1)}$ overhead \cite{Bodlaender1996LinearTimeTreewidth}.
Since $\Ncl{v}$ is a clique in $G^2$ for every $v\in V$, each local reward $f_v(x_{\Ncl{v}})$ can be evaluated inside a single bag.
Therefore, the relevant parameter is $\tw(G^2)$.

\subsection{Preliminaries: nice tree decompositions and cliques}
\label{subsec:alg-prelim}

Define $H\coloneqq G^2$.
We use a rooted \emph{nice} tree decomposition $(T,\{X_i\}_{i\in I})$ of $H$ with root $r$ and empty root bag, of width $t$, as in \cite[Ch.~7.2]{cygan2015parameterized}.

\begin{definition}[Nice tree decomposition]
\label{def:nice-td}
Let $H$ be an undirected graph.
A \emph{nice tree decomposition} of $H$ is a tree decomposition $(T,\{X_i\}_{i\in I})$ together with a root $r\in I$ such that $X_r=\emptyset$ and every leaf has empty bag.
Every node $i\in I$ is of exactly one of the following types.
\begin{itemize}
    \item \textbf{Leaf:} $i$ has no children and $X_i=\emptyset$.
    \item \textbf{Introduce:} $i$ has a single child $j$ and $X_i = X_j \cup \{a\}$ for some $a\notin X_j$.
    \item \textbf{Forget:} $i$ has a single child $j$ and $X_i = X_j \setminus \{a\}$ for some $a\in X_j$.
    \item \textbf{Join:} $i$ has exactly two children $j_1,j_2$ and $X_i = X_{j_1} = X_{j_2}$.
\end{itemize}
\end{definition}

It is standard that any width-$t$ tree decomposition can be converted in polynomial time into a nice tree decomposition of the same width and with $|I|=O(n\cdot t)$ nodes (see Ch. 7.2 of \cite{cygan2015parameterized}).
In particular, $|X_i|\le t+1$ for all $i\in I$.

\begin{lemma}[Cliques in bags \cite{Bodlaender2005DiscoveringTreewidth}]
\label{lem:clique-in-bag}
Let $(T,\{X_i\}_{i\in I})$ be a tree decomposition of $H$.
For every clique $W$ in $H$ there exists $i\in I$ such that $W\subseteq X_i$.
\end{lemma}

For every $v\in V$, the set $\Ncl{v}$ is a clique in $H$ since any two vertices in $\Ncl{v}$ have distance at most $2$ in $G$.
Therefore, by \Cref{lem:clique-in-bag}, for each $v$ there exists a bag $X_i$ with $\Ncl{v}\subseteq X_i$, so $f_v(x_{\Ncl{v}})$ can be evaluated from a single bag-labeling.

\subsection{Exact dynamic programming on $G^2$}
\label{subsec:cfdp}
We name the algorithm CFDP (Clique-Focused Dynamic Programming) to reflect the fact that it exploits the clique structure of the neighborhoods and the corresponding property of tree decompositions. Below we give a description of the algorithm; the correctness proof is deferred to \Cref{apd:proof-cfdp}.

\paragraph*{Assigning reward terms to bags}
Fix any mapping $\beta:V\to I$ such that $\Ncl{v}\subseteq X_{\beta(v)}$ for all $v\in V$.
Such a mapping can be constructed in polynomial time from $(T,\{X_i\}_{i\in I})$.
For each $i\in I$, let $V_i\coloneqq \{v\in V : \beta(v)=i\}$.
For a labeling $\ell:X_i\to\lcal$, define $\Phi_i(\ell) \coloneqq \sum_{v\in V_i} f_v\!\bigl(\ell_{\Ncl{v}}\bigr)$.
Then for every global labeling $x\in\lcal^V$ we have
\begin{equation}
\label{eq:cfdp-decomposition}
F(x) = \sum_{i\in I} \Phi_i\bigl(x_{X_i}\bigr).
\end{equation}

\paragraph*{Dynamic program}

For a node $i\in I$, let $T_i$ denote the set of nodes in the subtree of $T$ rooted at $i$, and let $U_i\coloneqq \bigcup_{j\in T_i} X_j$ be the set of vertices of $H$ appearing in bags of this subtree.
For a bag-labeling $\ell:X_i\to\lcal$, define $\mathrm{DP}_i(\ell) \coloneqq
\max\left\{\sum_{j\in T_i} \Phi_j\bigl(x_{X_j}\bigr) : x\in\lcal^{U_i},\; x_{X_i}=\ell\right\}$.
We process the decomposition bottom-up.
For each node type, the recurrence is as follows:
\[
\mathrm{DP}_i(\ell)
=
\begin{cases}
0, & \text{if $i$ is a leaf (so $X_i=\emptyset$ and $\ell=\emptyset$),}\\
\mathrm{DP}_j\bigl(\ell_{X_j}\bigr) + \Phi_i(\ell), & \text{if $i$ introduces $a$ and $X_i=X_j\cup\{a\}$,}\\
\max_{c\in\lcal}\;\mathrm{DP}_j\bigl(\ell\cup\{a\mapsto c\}\bigr) + \Phi_i(\ell), & \text{if $i$ forgets $a$ and $X_i=X_j\setminus\{a\}$,}\\
\mathrm{DP}_{j_1}(\ell) + \mathrm{DP}_{j_2}(\ell) + \Phi_i(\ell), & \text{if $i$ is a join node and $X_i=X_{j_1}=X_{j_2}$.}
\end{cases}
\]
Since the root bag is empty, the optimal value equals $\mathrm{DP}_r(\emptyset)$.
An optimal labeling can be recovered by standard backpointers at forget nodes.

\begin{theorem}
\label{thm:cfdp-correct}
There is an algorithm (CFDP) that, given an instance of \name on a graph $G=(V,E)$ with label set $\lcal$ ($L = |\lcal|$) and local reward functions $\{f_v\}_{v\in V}$, and a rooted nice tree decomposition $(T,\{X_i\}_{i\in I})$ of $H\coloneqq G^2$ of width $t$, computes a labeling $x^\star\in\lcal^V$ maximizing $F(x)$.
The running time is $O\bigl(n\cdot t\cdot L^{t+1}\bigr)$.
The space usage is $O(|I|\cdot L^{t+1})$.
If only the optimal value is required, the tables can be computed in a postorder traversal while discarding children after use, using $O(h\cdot L^{t+1})$ working space, where $h$ is the height of $T$.
\end{theorem}

\paragraph*{Running time intuition.}
Each bag has size at most $t+1$, hence there are at most $L^{t+1}$ states per node.
A standard implementation processes each node in $O(L^{t+1})$ time (including forget nodes via a single scan over child states) and a nice decomposition has $O(n\cdot t)$ nodes \cite[Ch.~7.2]{cygan2015parameterized}, yielding $O(n\cdot t\cdot L^{t+1})$ total time, up to polynomial factors in the input representation.
Derivations of the stated bounds and further implementation details appear in \Cref{subsec:cfdp-complexity}.

\subsection{Upper bound and tightness}
\label{subsec:alg-tight}

Given a width-$t$ tree decomposition of $H=G^2$, CFDP computes an optimal labeling in time $O\bigl(n\cdot t\cdot L^{t+1}\bigr)$.
In particular, taking $t=\tw(G^2)$ yields an algorithm running in time $O\bigl(n\cdot \tw(G^2)\cdot L^{\tw(G^2)+1}\bigr)$.
For binary labels $L=2$ this is $2^{O(\tw(G^2))}\cdot n^{O(1)}$, i.e., a single-exponential FPT algorithm parameterized by $\tw(G^2)$.
For constant $L$ this is $L^{O(\tw(G^2))}\cdot n^{O(1)}$.

\begin{remark}[Tightness under SETH]
Fix an integer $L\ge 2$.
By \Cref{thm:seTH-tight}, assuming SETH there is no algorithm for \name running in time $(L-\varepsilon)^{\tw(G^2)}\cdot n^{O(1)}$ for any $\varepsilon>0$.
Together with the $O\bigl(n\cdot \tw(G^2)\cdot L^{\tw(G^2)+1}\bigr)$ running time of CFDP, this shows that the base $L$ in the exponential dependence on $\tw(G^2)$ is optimal under SETH up to polynomial factors.
\end{remark}

\section{A coloring-based approximation for sparse graphs}
\label{sec:color-approx}
The inapproximability results of \Cref{subsec:inapprox} show that \name is NP-hard and hard to approximate on general graphs.
In many interference-bandit applications, reward values are nonnegative (often in $[0,1]$) and the interference graph is sparse, so each neighborhood $\Ncl{v}$ is small \cite{agarwal2024networkinterference,interference1-jia2024multi,jamshidi2025graph}.
Under such nonnegativity assumptions (and, when stating multiplicative ratios, assuming $\OPT>0$ so the guarantee is nonvacuous), meaningful approximation becomes possible.
In this section we give an approximation algorithm whose approximation ratio depends only on a coloring of the squared graph $G^2$.
In \Cref{sec:planar-ptas} we obtain a Baker-type PTAS for planar graphs.
Finally, \Cref{apd:submod} treats the budgeted monotone submodular variant.

Throughout this section, we assume that $f_v(\cdot) \ge 0$ for all $v \in V$. The following lemma states a key property of the color classes induced by a coloring of the squared graph $G^2$, which plays a central role in our algorithm.

\begin{lemma}[Disjoint neighborhoods within a $G^2$-independent set]
\label{lem:disjoint-neighborhoods}
Let $u,v\in V$ be distinct vertices. i.e.,  $\{u,v\}\notin E(G^2)$ (equivalently, $\dist_G(u,v)\ge 3$).
Then $\Ncl{u}\cap \Ncl{v}=\emptyset$.
Consequently, for any proper coloring $\psi$ of $G^2$, the sets $\{\Ncl{v}: \psi(v)=c\}$ are pairwise disjoint for every color $c$.
\end{lemma}

\begin{proof}
If $\Ncl{u}\cap \Ncl{v}\neq\emptyset$, then $\dist_G(u,v)\le 2$, hence $\{u,v\}\in E(G^2)$, a contradiction, and the second claim follows since each color class is independent in $G^2$.
\end{proof}

\paragraph*{Algorithm}
Let $H\coloneqq G^2$ and $\psi:V\to[q]$ be any proper coloring of $H$ with color classes $C_c\coloneqq \{v\in V:\psi(v)=c\}$.
\begin{enumerate}
    \item For each vertex $v$, compute $m_v \;\coloneqq\; \max\{f_v(a) \;|\; a:\Ncl{v}\to\lcal\}$ and fix a maximizer $a_v:\Ncl{v}\to\lcal$.
    (This can be done by enumerating all $L^{|\Ncl{v}|}$ assignments to $\Ncl{v}$.)
    \item For each color $c\in[q]$, assemble a labeling $x^{(c)}\in\lcal^V$ by setting, for every $v\in C_c$ and every $u\in\Ncl{v}$, $x^{(c)}_u \coloneqq a_v(u)$.
    Because the closed neighborhoods $\{\Ncl{v}:v\in C_c\}$ are disjoint (\Cref{lem:disjoint-neighborhoods}), no vertex receives conflicting labels.
    \item Assign an arbitrary default label to all remaining vertices, compute $M_c\coloneqq \sum_{v\in C_c} m_v$, and output $\hat x\coloneqq x^{(c^\star)}$ where $c^\star \;\in\; \argmax_{c\in[q]}\; M_c$.
\end{enumerate}

\begin{theorem}[Color-class maximization given a $G^2$-coloring]
\label{thm:color-approx}
Assume $f_v(\cdot)\ge 0$ for all $v\in V$.
Given a proper $q$-coloring $\psi$ of $G^2$, the above algorithm outputs a labeling $\hat x$ such that $F(\hat x)\;\ge\; \frac{1}{q}\cdot \OPT$.
The algorithm performs $\sum_{v\in V} L^{|\Ncl{v}|}$ oracle evaluations to compute $\{m_v\}$ and runs in time $O\left( L^{\Delta(G)+1}\right)\cdot n^{O(1)}$.
In particular, if $L$ is constant and $\Delta(G)$ is bounded, the algorithm runs in polynomial time.
\end{theorem}

\begin{proof}[Proof sketch]
Pick the color class maximizing $M_c$ and combine its disjoint neighborhood maximizers.
Full details (including the running-time bound) appear in \Cref{apd:color-proof}.
\end{proof}

\begin{remark}
    If a coloring is not provided, a greedy coloring of $G^2$ can be computed in polynomial time and uses at most $\Delta(G^2)+1$ colors \cite{west2001introduction}. This immediately yields a $1/(\Delta(G^2)+1)$-approximation. Moreover, since $\Delta(G^2) \leq \Delta(G)^2$, this guarantee is a constant-factor approximation when $\Delta(G)$ is bounded.
\end{remark}

\section{A Baker-type PTAS on planar bounded-degree graphs}
\label{sec:planar-ptas}

The coloring-based approximation of \Cref{sec:color-approx} yields a polynomial time $1/(\Delta(G^2) + 1)$-approximation.
On planar graphs of bounded maximum degree, one can do substantially better: although $G^2$ need not be planar and may have large treewidth, the special \emph{neighborhood} structure of \name (radius-$1$ dependencies in $G$) allows a Baker-style shifting scheme \cite{baker1994approximation}.

At a high level, we remove one BFS layer out of every $k$ layers, and for each offset optimize the objective induced by vertices whose closed neighborhoods avoid the removed layer.
We try all $k$ offsets and return the best resulting labeling.
A vertex can fail to be safe only if its layer or an adjacent layer is removed, so each vertex is unsafe for at most three offsets.
With nonnegative rewards, this implies that some offset retains at least a $(1-3/k)$ fraction of $\OPT$. We formalize this intuition below.

Throughout this section we assume that $f_v(\cdot)\ge 0$ for all $v\in V$, and that the input graph $G$ is planar with maximum degree $\Delta(G)$.

\paragraph*{Layering and shifting.}
Fix a root $r_0\in V$ and perform a breadth-first search (BFS) from $r_0$.
For each $v\in V$, let $\lambda(v)\coloneqq \dist_G(r_0,v)$, and define the BFS layers $B_i\coloneqq \{v\in V:\lambda(v)=i\}$ for $i\ge 0$.
For an integer $k\ge 3$ and an offset $s\in\{0,1,\ldots,k-1\}$, define the removed set
\[
R_s \;\coloneqq\; \{v\in V : \lambda(v)\equiv s \ (\mathrm{mod}\ k)\},
\qquad
G_s \;\coloneqq\; G[V\setminus R_s].
\]
Since $|\lambda(u)-\lambda(v)|\le 1$ for every edge $uv\in E$, removing $R_s$ splits $G_s$ into components contained in at most $k-1$ consecutive layers.
Define the set of \emph{safe} vertices
\[
S_s \;\coloneqq\; \{v\in V\setminus R_s : \Ncl{v}\cap R_s=\emptyset\}.
\]
For $v\in S_s$, we have $\Ncl{v}\subseteq V\setminus R_s$ by definition.
Hence the terms $\{f_v(x_{\Ncl{v}}): v\in S_s\}$ depend only on labels of $V\setminus R_s$.
The remaining vertices $(V\setminus R_s)\setminus S_s$ form a boundary whose local terms involve labels in $R_s$.

For each offset $s$, consider the modified objective
\[
F_s(x)\;\coloneqq\;\sum_{v\in S_s} f_v(x_{\Ncl{v}}),
\qquad x\in \lcal^{V\setminus R_s}.
\]
$F_s$ is obtained from $F$ by dropping exactly the terms that involve removed vertices.
We compute an \emph{optimal} labeling for $F_s$ on $G_s$ by exactly solving \name on each connected component of $G_s$ via CFDP (\Cref{thm:cfdp-correct}) after setting the local rewards of vertices outside $S_s$ to $0$.
We then extend the resulting partial labeling to all vertices in $R_s$ by assigning an arbitrary default label.
Finally, we evaluate the resulting $k$ labelings under $F$ and output the best one.

\begin{theorem}[Baker-type PTAS on planar bounded-degree graphs]
\label{thm:planar-ptas}
Assume $f_v(\cdot)\ge 0$ for all $v\in V$, and that $G$ is planar with maximum degree $\Delta(G)$.
For every $\varepsilon\in(0,1)$, there is an algorithm that outputs a labeling $\hat x\in \lcal^V$ such that $F(\hat x)\;\ge\;(1-\varepsilon)\cdot \OPT$, and the algorithm runs in time $O\!\left(n\cdot k^2 \cdot \Delta^2 \cdot L^{\,O \left( k\Delta(G)^2 \right)}\right)\cdot n^{O(1)}$, where $k\coloneqq \left\lceil \frac{3}{\varepsilon}\right\rceil$.
For constant $L$ and $\Delta(G)$, this is an EPTAS.
\end{theorem}

\begin{proof}
Set $k\coloneqq \lceil 3/\varepsilon\rceil$ and consider the $k$ offsets $s\in\{0,1,\ldots,k-1\}$.
We first show that some offset preserves almost all of $\OPT$ on the safe vertices, and then argue that our algorithm finds a labeling achieving at least that value.

Let $x^\star$ be an optimal labeling for the original instance.
For each offset $s$, define $\mathrm{OPT}_s \;\coloneqq\; \sum_{v\in S_s} f_v\!\bigl(x^\star_{\Ncl{v}}\bigr)$.
That is, $\mathrm{OPT}_s$ is the contribution of vertices that are safe under offset $s$ in the optimal labeling $x^\star$.
We claim that each vertex $v\in V$ belongs to $S_s$ for all but at most three offsets $s$.
Indeed, $v\notin S_s$ iff $v\in R_s$ or $v$ has a neighbor in $R_s$.
Since $\lambda(\cdot)$ is a BFS layering, any neighbor $u$ of $v$ satisfies $|\lambda(u)-\lambda(v)|\le 1$.
Thus $v\notin S_s$ can occur only when $s\equiv \lambda(v)$, $s\equiv \lambda(v)-1$, or $s\equiv \lambda(v)+1\pmod{k}$, i.e., for at most three offsets.

Since all local rewards are nonnegative,
\[
\sum_{s=0}^{k-1} \mathrm{OPT}_s
\;\ge\;
\sum_{v\in V} (k-3)\cdot f_v\!\bigl(x^\star_{\Ncl{v}}\bigr)
\;=\;
(k-3)\cdot \OPT,
\]
and hence there exists an offset $s$ such that
\[
\mathrm{OPT}_s \;\ge\; \left(1-\frac{3}{k}\right)\OPT \;\ge\; (1-\varepsilon)\OPT.
\]
Fix such an offset $s$.

By definition of $S_s$, for every $v\in S_s$ we have $\Ncl{v}\subseteq V\setminus R_s$.
Therefore the restriction $x^\star_{V\setminus R_s}$ is feasible for maximizing $F_s(\cdot)$ and achieves value exactly $\mathrm{OPT}_s$.
Consequently, the optimal value of $F_s$ on $G_s$ is at least $\mathrm{OPT}_s$.
Our algorithm computes an optimal labeling $x^{(s)}$ for $F_s$ (by solving each component exactly and combining), hence
\[
F_s\!\left(x^{(s)}\right)\;\ge\;\mathrm{OPT}_s.
\]
Extending $x^{(s)}$ arbitrarily to $R_s$ yields a full labeling $\hat x$.
Since all rewards are nonnegative, $F(\hat x)\ge F_s(\hat x)=F_s(x^{(s)})\ge \mathrm{OPT}_s\ge (1-\varepsilon)\OPT$.

Fix an offset $s$ and a connected component $C$ of $G_s$.
Each such component lies between two removed layers, so it has bounded depth in the BFS layering.
Contracting the BFS layers closer to $r_0$ than $C$ (in particular the removed layer adjacent to $C$) to a single vertex yields a planar graph in which $C$ has radius $O(k)$, and hence $\tw(C)=O(k)$ \cite{eppstein2000diameter}.
Using the bag-expansion bound from \Cref{subsec:twG2} yields $\tw(C^2)=O \left( k\Delta(G)^2 \right)$.
Thus CFDP runs in $O\!\left(|C|\cdot k\Delta(G)^2\cdot L^{O \left( k\Delta(G)^2 \right)}\right)\cdot n^{O(1)}$ time on each component, and summing over components and the $k$ offsets yields the stated bound.
Further details appear in \Cref{apd:ptas-runtime}.
\end{proof}


\section{Conclusion and open problems}
\label{sec:conclusion}
We studied the Neighborhood-Aware Graph Labeling (\name) problem, motivated by optimization oracles in bandits with network interference.
We gave a standard dynamic program over a tree decomposition of $G^2$ and showed that, under SETH, the base $L$ in the exponential dependence on $\tw(G^2)$ is optimal up to polynomial factors.
For approximation with nonnegative rewards, we derived coloring-based algorithms and a Baker-type PTAS on planar graphs of bounded maximum degree (yielding an EPTAS for constant $L$).
Open problems include developing approximation guarantees driven by structure in the reward functions rather than by the graph class.
In particular, it would be natural to study the restricted reward-function families that arise in the network-interference literature (e.g., exposure-mapping or other low-complexity neighborhood summaries) and ask which such families admit polynomial-time approximation algorithms, possibly even on general graphs.

\clearpage
\phantomsection
\label{pg:bib-start}
\bibliography{biblo}

@inproceedings{Bodlaender2005DiscoveringTreewidth,
  author    = {Hans L. Bodlaender},
  title     = {Discovering Treewidth},
  booktitle = {SOFSEM 2005: Theory and Practice of Computer Science ({SOFSEM} 2005)},
  series    = {Lecture Notes in Computer Science},
  volume    = {3381},
  pages     = {1--16},
  publisher = {Springer},
  year      = {2005},
  doi       = {10.1007/978-3-540-30577-4_1},
}

@article{Impagliazzo2001kSAT,
  author  = {Russell Impagliazzo and Ramamohan Paturi},
  title   = {On the Complexity of {k}-{SAT}},
  journal = {Journal of Computer and System Sciences},
  volume  = {62},
  number  = {2},
  pages   = {367--375},
  year    = {2001},
  doi     = {10.1006/jcss.2000.1727},
}

@inproceedings{agarwal2024networkinterference,
  author    = {Abhineet Agarwal and Anish Agarwal and Lorenzo Masoero and Justin Whitehouse},
  title     = {Multi-Armed Bandits with Network Interference},
  booktitle = {Advances in Neural Information Processing Systems 37 ({NeurIPS} 2024)},
  pages     = {36414--36437},
  year      = {2024},
  doi       = {10.52202/079017-1148},
}

@inproceedings{bodlaender1993linear,
  author    = {Hans L. Bodlaender},
  title     = {A linear time algorithm for finding tree-decompositions of small treewidth},
  booktitle = {Proceedings of the Twenty-Fifth Annual {ACM} Symposium on Theory of Computing ({STOC} '93)},
  pages     = {226--234},
  publisher = {ACM},
  year      = {1993},
  doi       = {10.1145/167088.167161},
}

@article{boykov2001graphcuts,
  author  = {Yuri Boykov and Olga Veksler and Ramin Zabih},
  title   = {Fast approximate energy minimization via graph cuts},
  journal = {IEEE Transactions on Pattern Analysis and Machine Intelligence},
  volume  = {23},
  number  = {11},
  pages   = {1222--1239},
  year    = {2001},
  doi     = {10.1109/34.969114},
}

@article{cesabianchi2012combinatorial,
  author  = {Nicol{\`o} Cesa-Bianchi and G{\'a}bor Lugosi},
  title   = {Combinatorial bandits},
  journal = {Journal of Computer and System Sciences},
  volume  = {78},
  number  = {5},
  pages   = {1404--1422},
  year    = {2012},
  doi     = {10.1016/j.jcss.2012.01.001},
}

@inproceedings{chen2013combinatorial,
  author    = {Wei Chen and Yajun Wang and Yang Yuan},
  title     = {Combinatorial multi-armed bandit: General framework and applications},
  booktitle = {Proceedings of the 30th International Conference on Machine Learning ({ICML} 2013)},
  series    = {JMLR Workshop and Conference Proceedings},
  volume    = {28},
  pages     = {151--159},
  publisher = {JMLR.org},
  year      = {2013},
}

@inproceedings{combes2015combinatorial,
  author    = {Richard Combes and Mohammad Sadegh Talebi and Alexandre Prouti{\`{e}}re and Marc Lelarge},
  title     = {Combinatorial bandits revisited},
  booktitle = {Advances in Neural Information Processing Systems 28 ({NeurIPS} 2015)},
  pages     = {2116--2124},
  year      = {2015},
}

@book{cygan2015parameterized,
  author    = {Marek Cygan and Fedor V. Fomin and {\L}ukasz Kowalik and Daniel Lokshtanov and D{\'a}niel Marx and Marcin Pilipczuk and Micha{\l} Pilipczuk and Saket Saurabh},
  title     = {Parameterized Algorithms},
  publisher = {Springer},
  year      = {2015},
  doi       = {10.1007/978-3-319-21275-3},
}

@article{dechter1989treeclustering,
  author  = {Rina Dechter and Judea Pearl},
  title   = {Tree clustering for constraint networks},
  journal = {Artificial Intelligence},
  volume  = {38},
  number  = {3},
  pages   = {353--366},
  year    = {1989},
  doi     = {10.1016/0004-3702(89)90037-4},
}

@inproceedings{freuder1990ktree,
  author    = {Eugene C. Freuder},
  title     = {Complexity of {K}-Tree Structured Constraint Satisfaction Problems},
  booktitle = {Proceedings of the 8th National Conference on Artificial Intelligence ({AAAI}-90)},
  pages     = {4--9},
  publisher = {AAAI Press / The MIT Press},
  year      = {1990},
}

@article{general-factor-arulselvan2018matchings,
  author  = {Ashwin Arulselvan and {\'A}gnes Cseh and Martin Gro{\ss} and David F. Manlove and Jannik Matuschke},
  title   = {Matchings with lower quotas: Algorithms and complexity},
  journal = {Algorithmica},
  volume  = {80},
  number  = {1},
  pages   = {185--208},
  year    = {2018},
  doi     = {10.1007/s00453-016-0252-6},
}

@inproceedings{general-factor2-marx2021degrees,
  author    = {D{\'a}niel Marx and Govind S. Sankar and Philipp Schepper},
  title     = {Degrees and Gaps: Tight Complexity Results of General Factor Problems Parameterized by Treewidth and Cutwidth},
  booktitle = {48th International Colloquium on Automata, Languages, and Programming ({ICALP} 2021)},
  series    = {LIPIcs},
  volume    = {198},
  pages     = {95:1--95:20},
  year      = {2021},
  doi       = {10.4230/LIPIcs.ICALP.2021.95},
}

@inproceedings{interference1-jia2024multi,
  author    = {Su Jia and Peter I. Frazier and Nathan Kallus},
  title     = {Multi-Armed Bandits with Interference: Bridging Causal Inference and Adversarial Bandits},
  booktitle = {Proceedings of the 42nd International Conference on Machine Learning ({ICML} 2025)},
  series    = {Proceedings of Machine Learning Research},
  volume    = {267},
  pages     = {27271--27291},
  publisher = {PMLR},
  year      = {2025},
}

@misc{interference2-faruk2025learning,
      title={Learning Peer Influence Probabilities with Linear Contextual Bandits}, 
      author={Ahmed Sayeed Faruk and Mohammad Shahverdikondori and Elena Zheleva},
      year={2025},
      eprint={2510.19119},
      archivePrefix={arXiv},
      primaryClass={cs.LG},
      url={https://arxiv.org/abs/2510.19119}, 
}

@misc{jamshidi2025graph,
      title={Graph-Dependent Regret Bounds in Multi-Armed Bandits with Interference}, 
      author={Fateme Jamshidi and Mohammad Shahverdikondori and Negar Kiyavash},
      year={2025},
      eprint={2503.07555},
      archivePrefix={arXiv},
      primaryClass={cs.LG},
      url={https://arxiv.org/abs/2503.07555}, 
}

@inproceedings{kempe2003maximizing,
  author    = {David Kempe and Jon M. Kleinberg and {\'E}va Tardos},
  title     = {Maximizing the spread of influence through a social network},
  booktitle = {Proceedings of the Ninth {ACM} {SIGKDD} International Conference on Knowledge Discovery and Data Mining ({KDD} 2003)},
  pages     = {137--146},
  publisher = {ACM},
  year      = {2003},
  doi       = {10.1145/956750.956769},
}

@book{koller2009pgm,
  author    = {Daphne Koller and Nir Friedman},
  title     = {Probabilistic Graphical Models: Principles and Techniques},
  publisher = {MIT Press},
  year      = {2009},
}

@book{west2001introduction,
  author    = {Douglas B. West},
  title     = {Introduction to Graph Theory},
  edition   = {2},
  publisher = {Prentice Hall},
  year      = {2001},
}

@article{kolmogorov2004graphcuts,
  author  = {Vladimir Kolmogorov and Ramin Zabih},
  title   = {What energy functions can be minimized via graph cuts?},
  journal = {IEEE Transactions on Pattern Analysis and Machine Intelligence},
  volume  = {26},
  number  = {2},
  pages   = {147--159},
  year    = {2004},
  doi     = {10.1109/TPAMI.2004.1262177},
}

@incollection{krause2014submodular,
  author    = {Andreas Krause and Daniel Golovin},
  title     = {Submodular function maximization},
  booktitle = {Tractability: Practical Approaches to Hard Problems},
  pages     = {71--104},
  publisher = {Cambridge University Press},
  year      = {2014},
  doi       = {10.1017/CBO9781139177801.004},
}

@article{lauritzen1988local,
  author  = {Steffen L. Lauritzen and David J. Spiegelhalter},
  title   = {Local computations with probabilities on graphical structures and their application to expert systems},
  journal = {Journal of the Royal Statistical Society: Series B (Methodological)},
  volume  = {50},
  number  = {2},
  pages   = {157--194},
  year    = {1988},
  doi     = {10.1111/j.2517-6161.1988.tb01721.x},
}

@inproceedings{lokshtanovMarxSaurabhSODA11,
  author    = {Daniel Lokshtanov and D{\'a}niel Marx and Saket Saurabh},
  title     = {Known Algorithms on Graphs of Bounded Treewidth are Probably Optimal},
  booktitle = {Proceedings of the Twenty-Second Annual {ACM}-{SIAM} Symposium on Discrete Algorithms ({SODA} 2011)},
  pages     = {777--789},
  publisher = {SIAM},
  year      = {2011},
  doi       = {10.1137/1.9781611973082.61},
}

@incollection{lower-treewidth-cygan2015lower,
  author    = {Marek Cygan and Fedor V. Fomin and {\L}ukasz Kowalik and Daniel Lokshtanov and D{\'a}niel Marx and Marcin Pilipczuk and Micha{\l} Pilipczuk and Saket Saurabh},
  title     = {Lower bounds based on the exponential-time hypothesis},
  booktitle = {Parameterized Algorithms},
  pages     = {467--521},
  publisher = {Springer},
  year      = {2015},
  doi       = {10.1007/978-3-319-21275-3_14},
}

@inproceedings{min-stable-cut-lampis2021minimum,
  author    = {Michael Lampis},
  title     = {Minimum Stable Cut and Treewidth},
  booktitle = {48th International Colloquium on Automata, Languages, and Programming ({ICALP} 2021)},
  series    = {LIPIcs},
  volume    = {198},
  pages     = {92:1--92:16},
  year      = {2021},
  doi       = {10.4230/LIPIcs.ICALP.2021.92},
}

@inproceedings{mis-hardness-zuckerman2006linear,
  author    = {David Zuckerman},
  title     = {Linear degree extractors and the inapproximability of max clique and chromatic number},
  booktitle = {Proceedings of the Thirty-Eighth Annual {ACM} Symposium on Theory of Computing ({STOC} '06)},
  pages     = {681--690},
  publisher = {ACM},
  year      = {2006},
  doi       = {10.1145/1132516.1132612},
}

@article{nemhauser_1978,
  author  = {George L. Nemhauser and Laurence A. Wolsey and Marshall L. Fisher},
  title   = {An analysis of approximations for maximizing submodular set functions---I},
  journal = {Mathematical Programming},
  volume  = {14},
  number  = {1},
  pages   = {265--294},
  year    = {1978},
  doi     = {10.1007/BF01588971},
}

@article{power-treewidth-gurski2025behavior,
  author  = {Frank Gurski and Robin Weishaupt},
  title   = {The Behavior of Tree-Width and Path-Width Under Graph Operations and Graph Transformations},
  journal = {Algorithms},
  volume  = {18},
  number  = {7},
  pages   = {386},
  year    = {2025},
  doi     = {10.3390/a18070386},
}

@inproceedings{rossi2015networkrepository,
  author    = {Ryan A. Rossi and Nesreen K. Ahmed},
  title     = {The Network Data Repository with Interactive Graph Analytics and Visualization},
  booktitle = {Proceedings of the Twenty-Ninth {AAAI} Conference on Artificial Intelligence ({AAAI} 2015)},
  pages     = {4292--4293},
  publisher = {AAAI Press},
  year      = {2015},
  doi       = {10.1609/AAAI.V29I1.9277},
}

@inproceedings{tamaki2017positiveinstance,
  author    = {Hisao Tamaki},
  title     = {Positive-Instance Driven Dynamic Programming for Treewidth},
  booktitle = {25th Annual European Symposium on Algorithms ({ESA} 2017)},
  series    = {LIPIcs},
  volume    = {87},
  pages     = {68:1--68:13},
  year      = {2017},
  doi       = {10.4230/LIPIcs.ESA.2017.68},
}

@incollection{treewidth-survey-bienstock1991graph,
  author    = {Daniel Bienstock},
  title     = {Graph searching, path-width, tree-width and related problems (a survey)},
  booktitle = {Reliability of Computer and Communication Networks},
  series    = {DIMACS Series in Discrete Mathematics and Theoretical Computer Science},
  volume    = {5},
  pages     = {33--50},
  publisher = {American Mathematical Society},
  year      = {1991},
  doi       = {10.1090/dimacs/005/02},
}

@article{tw-small-delta-grohe2009tree,
  author  = {Martin Grohe and D{\'a}niel Marx},
  title   = {On tree width, bramble size, and expansion},
  journal = {Journal of Combinatorial Theory, Series B},
  volume  = {99},
  number  = {1},
  pages   = {218--228},
  year    = {2009},
  doi     = {10.1016/j.jctb.2008.06.004},
}

@article{tw-small-delta2-feige2016giant,
  author  = {Uriel Feige and Jonathan Hermon and Daniel Reichman},
  title   = {On giant components and treewidth in the layers model},
  journal = {Random Structures \& Algorithms},
  volume  = {48},
  number  = {3},
  pages   = {524--545},
  year    = {2016},
  doi     = {10.1002/rsa.20597},
}

@article{wainwright2008graphical,
  author  = {Martin J. Wainwright and Michael I. Jordan},
  title   = {Graphical Models, Exponential Families, and Variational Inference},
  journal = {Foundations and Trends in Machine Learning},
  volume  = {1},
  number  = {1--2},
  pages   = {1--305},
  year    = {2008},
  doi     = {10.1561/2200000001},
}

@article{aronowSamii2017interference,
  author  = {Peter M. Aronow and Cyrus Samii},
  title   = {Estimating Average Causal Effects under General Interference, with Application to a Social Network Experiment},
  journal = {The Annals of Applied Statistics},
  volume  = {11},
  number  = {4},
  pages   = {1912--1947},
  year    = {2017},
  doi     = {10.1214/16-AOAS1005},
}

@article{Bodlaender1996LinearTimeTreewidth,
  author  = {Hans L. Bodlaender},
  title   = {A Linear-Time Algorithm for Finding Tree-Decompositions of Small Treewidth},
  journal = {SIAM Journal on Computing},
  volume  = {25},
  number  = {6},
  pages   = {1305--1317},
  year    = {1996},
  doi     = {10.1137/S0097539793251219},
}

@article{baker1994approximation,
  author  = {Brenda S. Baker},
  title   = {Approximation Algorithms for {NP}-Complete Problems on Planar Graphs},
  journal = {Journal of the ACM},
  volume  = {41},
  number  = {1},
  pages   = {153--180},
  year    = {1994},
  doi     = {10.1145/174644.174650},
}

@article{eppstein2000diameter,
  author  = {David Eppstein},
  title   = {Diameter and Treewidth in Minor-Closed Graph Families},
  journal = {Algorithmica},
  volume  = {27},
  number  = {3--4},
  pages   = {275--291},
  year    = {2000},
  doi     = {10.1007/s004530010020},
}

@article{elahi2023contextual,
  author  = {Sepehr Elahi and Baran Atalar and Sevda {\"O}{\u{g}}{\"u}t and Cem Tekin},
  title   = {Contextual Combinatorial Multi-output {GP} Bandits with Group Constraints},
  journal = {Transactions on Machine Learning Research},
  year    = {2023},
}

@inproceedings{comb_bandit2020,
  author    = {Andi Nika and Sepehr Elahi and Cem Tekin},
  title     = {Contextual Combinatorial Volatile Multi-armed Bandit with Adaptive Discretization},
  booktitle = {Proceedings of the Twenty Third International Conference on Artificial Intelligence and Statistics ({AISTATS} 2020)},
  series    = {Proceedings of Machine Learning Research},
  volume    = {108},
  pages     = {1486--1496},
  publisher = {PMLR},
  year      = {2020},
}

@inproceedings{elahi2025vacc,
  author    = {Sepehr Elahi and Paula M{\"u}rmann and Patrick Thiran},
  title     = {Learn to Vaccinate: Combining Structure Learning and Effective Vaccination for Epidemic and Outbreak Control},
  booktitle = {Proceedings of the 42nd International Conference on Machine Learning ({ICML} 2025)},
  series    = {Proceedings of Machine Learning Research},
  volume    = {267},
  pages     = {15182--15214},
  publisher = {PMLR},
  year      = {2025},
}

@article{shahverdikondori2025graph,
  title={Graph Learning is Suboptimal in Causal Bandits},
  author={Shahverdikondori, Mohammad and Etesami, Jalal and Kiyavash, Negar},
  journal={arXiv preprint arXiv:2510.16811},
  year={2025}
}

\clearpage
\phantomsection
\label{pg:appendix-start}

\crefname{appendix}{Appendix}{Appendices}
\Crefname{appendix}{Appendix}{Appendices}

\appendix
\crefalias{section}{appendix}
\crefalias{subsection}{appendix}

\begin{center} \label{sec: apd}
    {\Large \textbf{Appendices}}
\end{center}

\section{Further discussion on related work} \label{app:related_work}

\paragraph*{Treewidth-based exact algorithms and dynamic programming}
Dynamic programming on tree decompositions is a central technique for designing exact algorithms for NP-hard graph problems, with running times typically exponential in the treewidth of the input graph.
Classic examples include \textsc{Maximum Independent Set}, \textsc{Vertex Cover}, \textsc{General Factor}, and various connectivity problems (e.g., Hamiltonicity variants), all of which admit treewidth-parameterized algorithms via bag-based DP (see, e.g., \cite{cygan2015parameterized,bodlaender1993linear,treewidth-survey-bienstock1991graph,Bodlaender2005DiscoveringTreewidth,general-factor-arulselvan2018matchings,general-factor2-marx2021degrees}).
ETH- and SETH-based lower bounds further suggest that, for many problems, this exponential dependence on treewidth is essentially unavoidable \cite{lower-treewidth-cygan2015lower,cygan2015parameterized,lokshtanovMarxSaurabhSODA11}.
In contrast, the local-interaction structure of \name makes the squared graph $G^2$ the relevant object for complexity analysis.
As discussed in \Cref{subsec:twG2}, the problem cannot be made tractable by assuming that either $\tw(G)$ or $\Delta(G)$ is constant.
Instead, tractability requires both quantities to be small.
Similar phenomena have been observed in prior work.
One example is the \textsc{Minimum Stable Cut} problem, which is shown to admit algorithms with running time exponential in $\tw(G)\cdot\Delta(G)$ \cite{min-stable-cut-lampis2021minimum}.
Although that problem is also a combinatorial optimization problem with neighborhood-based objective functions, its underlying source of hardness differs from that of \name, since in our setting $\tw(G^2)$ can be significantly smaller than $\tw(G)\cdot\Delta(G)$.
A recent epidemic-control formulation with local interactions also exploits treewidth-parameterized optimization methods; see \cite{elahi2025vacc}.

\paragraph*{Graph labeling and local interaction models}
A large body of work studies graph labeling objectives in which the global score decomposes into unary and pairwise terms.
\begin{equation*}
\sum_{v\in V} g_v(x_v) \;+\; \sum_{(u,v)\in E}g_{uv}(x_u,x_v).
\end{equation*}
This covers standard formulations in networked resource allocation and MAP inference in pairwise Markov random fields (MRFs), and is a dominant modeling choice in computer vision tasks such as image segmentation and stereo matching \cite{boykov2001graphcuts,kolmogorov2004graphcuts,wainwright2008graphical,koller2009pgm}.
For these pairwise models, exact inference and MAP can be performed by variable elimination and junction-tree style algorithms whose time and space complexity are exponential in the treewidth of the \emph{original} graph $G$, with typical bounds of the form $O\!\left(n\cdot L^{\tw(G)+1}\right)$ for label alphabet size $L$ \cite{lauritzen1988local,dechter1989treeclustering}.
Equivalently, \name can be viewed as a weighted constraint satisfaction problem (or factor graph) with one factor whose scope is $\Ncl{v}$ for each $v\in V$.
The corresponding primal graph is exactly $G^2$ \cite{dechter1989treeclustering,freuder1990ktree}.
Our work can be seen as a strict generalization of the unary/pairwise setting.
In particular, \name allows \emph{arbitrary} neighborhood reward functions $f_v(x_{\Ncl{v}})$, which in general cannot be represented as sums over edges.
This leads to a fundamentally different structural parameter, namely $\tw(G^2)$.

\paragraph*{Combinatorial bandits and bandits with interference (oracle motivation)}
A recurring theme in multi-armed bandit models in sequential learning with large or structured action spaces is the use of an \emph{offline optimization oracle}.
Given estimated rewards (or upper confidence bounds), the learner repeatedly selects an action maximizing a global objective over an exponentially large decision set.
This paradigm is standard in combinatorial bandits \cite{cesabianchi2012combinatorial,chen2013combinatorial,combes2015combinatorial, comb_bandit2020, elahi2023contextual, shahverdikondori2025graph}.
It becomes particularly salient in multi-armed bandits with \emph{network interference} \cite{interference1-jia2024multi,agarwal2024networkinterference,jamshidi2025graph}, where each round assigns an arm (label) to every unit (node) and each unit's reward depends on the joint assignment within its closed neighborhood.
Closely related but distinct, \cite{interference2-faruk2025learning} studies learning peer influence probabilities under a contextual bandit model with edge-level interventions.
Interference-aware algorithms typically maintain local models for each node as a function of its neighborhood assignment.
They then require maximizing the sum of these local objectives to choose a globally consistent joint assignment.
This offline step coincides exactly with \name (with labels corresponding to arms and node rewards given by the current local estimates).
This oracle is made explicit, for example, in the PUCB-I algorithm of \cite{jamshidi2025graph} and in the elimination-based approach of \cite{agarwal2024networkinterference}, where repeated global maximization is used as a subroutine across epochs.
Our results complement this line of work by characterizing the computational and approximability limits of this oracle and by providing an optimal exact algorithm under bounded $\tw(G^2)$.

\section{Omitted proofs} \label{apd:proofs}

This section presents the omitted proofs of the main text.

\subsection{Proof of Theorem~\ref{thm:cfdp-correct}}
\label{apd:proof-cfdp}

\begin{proof}
We prove the correctness of CFDP.
The running time and space bounds are analyzed in \Cref{subsec:cfdp-complexity}.

Let $H\coloneqq G^2$.
Let $(T,\{X_i\}_{i\in I})$ be the rooted nice tree decomposition of $H$ given to the algorithm, with root $r$.
For every $v\in V$, recall that $\Ncl{v}$ denotes the closed neighborhood of $v$ in $G$.

By \Cref{lem:clique-in-bag}, for every $v\in V$ there exists a bag containing $\Ncl{v}$.
Fix any mapping $\beta:V\to I$ such that $\Ncl{v}\subseteq X_{\beta(v)}$ for all $v\in V$.
For each $i\in I$, define $V_i\coloneqq \{v\in V : \beta(v)=i\}$.
For a labeling $\ell:X_i\to\lcal$, recall $\Phi_i(\ell)\coloneqq \sum_{v\in V_i} f_v\!\bigl(\ell_{\Ncl{v}}\bigr)$.
By \eqref{eq:cfdp-decomposition}, for every labeling $x\in\lcal^V$ we have $F(x)=\sum_{i\in I} \Phi_i\bigl(x_{X_i}\bigr)$.

For a node $i\in I$, recall that $T_i$ denotes the set of nodes in the subtree of $T$ rooted at $i$.
Define $U_i\coloneqq \bigcup_{j\in T_i} X_j$.
For fixed $i$ and $\ell:X_i\to\lcal$, the quantity $\sum_{j\in T_i}\Phi_j\bigl(x_{X_j}\bigr)$ depends only on the restriction $x_{U_i}$.
Moreover, any labeling of $U_i$ consistent with $\ell$ can be extended arbitrarily to a labeling of $V$ without affecting $\sum_{j\in T_i}\Phi_j(x_{X_j})$.
Therefore,
\begin{equation}
\label{eq:opt-local-domain}
\max\left\{\sum_{j\in T_i}\Phi_j\bigl(x_{X_j}\bigr) : x\in \lcal^V,\; x_{X_i}=\ell\right\}
=
\max\left\{\sum_{j\in T_i}\Phi_j\bigl(x_{X_j}\bigr) : x\in \lcal^{U_i},\; x_{X_i}=\ell\right\}.
\end{equation}
Let $\mathrm{OPT}_i(\ell)$ denote this common value.

We prove by induction on the decomposition that for every $i\in I$ and every $\ell:X_i\to\lcal$, the value computed by CFDP satisfies $\mathrm{DP}_i(\ell)=\mathrm{OPT}_i(\ell)$.

\paragraph*{Leaf.}
Let $i$ be a leaf.
Then $X_i=\emptyset$.
Since $\Ncl{v}\neq\emptyset$ for every $v\in V$, we have $V_i=\emptyset$ and thus $\Phi_i(\emptyset)=0$.
Hence $\mathrm{OPT}_i(\emptyset)=0=\mathrm{DP}_i(\emptyset)$.

\paragraph*{Introduce.}
Let $i$ be an introduce node with child $j$ and $X_i=X_j\cup\{a\}$.
We claim that $a\notin U_j$.
Indeed, if $a\in X_k$ for some $k\in T_j$, then by the connectedness condition of tree decompositions the vertex $a$ must also belong to $X_j$, contradicting $a\notin X_j$.
Consequently, in the maximization defining $\mathrm{OPT}_i(\ell)$, the contribution of the subtree $T_j$ depends on $\ell$ only through its restriction $\ell_{X_j}$, and its optimal value is $\mathrm{OPT}_j(\ell_{X_j})$.
Using $T_i=T_j\cup\{i\}$, we obtain
\begin{equation*}
\mathrm{OPT}_i(\ell)=\Phi_i(\ell)+\mathrm{OPT}_j\bigl(\ell_{X_j}\bigr),
\end{equation*}
which matches the introduce recurrence.
By the induction hypothesis, $\mathrm{DP}_j\bigl(\ell_{X_j}\bigr)=\mathrm{OPT}_j\bigl(\ell_{X_j}\bigr)$.
Therefore $\mathrm{DP}_i(\ell)=\mathrm{OPT}_i(\ell)$.

\paragraph*{Forget.}
Let $i$ be a forget node with child $j$ and $X_i=X_j\setminus\{a\}$.
Fix $\ell:X_i\to\lcal$.
For any $x\in\lcal^{U_i}$ with $x_{X_i}=\ell$, let $c\coloneqq x_a$ and define $\ell'\coloneqq \ell\cup\{a\mapsto c\}$.
Then $\ell':X_j\to\lcal$ and $x_{X_j}=\ell'$.
Since $T_i=T_j\cup\{i\}$ and $\Phi_i$ depends only on labels on $X_i$, we obtain
\begin{equation*}
\mathrm{OPT}_i(\ell)=\Phi_i(\ell)+\max_{c\in\lcal}\;\mathrm{OPT}_j\bigl(\ell\cup\{a\mapsto c\}\bigr),
\end{equation*}
which matches the forget recurrence.
By the induction hypothesis, this implies $\mathrm{DP}_i(\ell)=\mathrm{OPT}_i(\ell)$.

\paragraph*{Join.}
Let $i$ be a join node with children $j_1,j_2$ and $X_i=X_{j_1}=X_{j_2}$.
We claim that $U_{j_1}\cap U_{j_2}=X_i$.
If a vertex $v$ belongs to both $U_{j_1}$ and $U_{j_2}$, then it appears in some bag in each of the two subtrees.
By the connectedness condition, every bag on the path between these occurrences contains $v$, in particular $X_i$ contains $v$.
Thus $U_{j_1}\cap U_{j_2}\subseteq X_i$.
The reverse inclusion is immediate since $X_i\subseteq U_{j_1}\cap U_{j_2}$.
Moreover, $U_i=U_{j_1}\cup U_{j_2}$.

Fix a labeling $\ell:X_i\to\lcal$.
Since $U_{j_1}\cap U_{j_2}=X_i$, any pair of labelings $x_1\in\lcal^{U_{j_1}}$ and $x_2\in\lcal^{U_{j_2}}$ with ${x_1}_{X_i}={x_2}_{X_i}=\ell$ can be combined into a labeling $x\in\lcal^{U_i}$.
Moreover, for such an $x$ we have
$\sum_{k\in T_{j_1}}\Phi_k\bigl(x_{X_k}\bigr)=\sum_{k\in T_{j_1}}\Phi_k\bigl((x_1)_{X_k}\bigr)$
and
$\sum_{k\in T_{j_2}}\Phi_k\bigl(x_{X_k}\bigr)=\sum_{k\in T_{j_2}}\Phi_k\bigl((x_2)_{X_k}\bigr)$.
It follows that
\begin{equation*}
\max_{x\in\lcal^{U_i},\,x_{X_i}=\ell}
\left(\sum_{k\in T_{j_1}}\Phi_k\bigl(x_{X_k}\bigr)+\sum_{k\in T_{j_2}}\Phi_k\bigl(x_{X_k}\bigr)\right)
=\mathrm{OPT}_{j_1}(\ell)+\mathrm{OPT}_{j_2}(\ell).
\end{equation*}
Using $T_i=T_{j_1}\cup T_{j_2}\cup\{i\}$, we obtain
\begin{equation*}
\mathrm{OPT}_i(\ell)=\Phi_i(\ell)+\mathrm{OPT}_{j_1}(\ell)+\mathrm{OPT}_{j_2}(\ell),
\end{equation*}
which matches the join recurrence.
By the induction hypothesis, this implies $\mathrm{DP}_i(\ell)=\mathrm{OPT}_i(\ell)$.

This completes the induction. Finally, since $X_r=\emptyset$ and $U_r=V$, we have $\mathrm{DP}_r(\emptyset)=\mathrm{OPT}_r(\emptyset)=\max_{x\in\lcal^V}\;\sum_{i\in I}\Phi_i\bigl(x_{X_i}\bigr)$.
By \eqref{eq:cfdp-decomposition}, this equals $\max_{x\in\lcal^V} F(x)$.
CFDP therefore computes the optimal value.
Moreover, storing the maximizing choice at each forget step yields a consistent set of backpointers. At join nodes we recurse into both children with the same bag labeling $\ell$, and consistency is ensured by the shared interface $X_i$.
Following these backpointers reconstructs a labeling $x^\star\in\lcal^V$ achieving $\mathrm{DP}_r(\emptyset)$.
Thus CFDP outputs an optimal labeling for \name.
\end{proof}

\subsection{Running time and space of CFDP}
\label{subsec:cfdp-complexity}

We justify the bounds stated in \Cref{thm:cfdp-correct}.
Let $t$ be the width of the decomposition of $G^2$.
Each bag has size at most $t+1$, hence each DP table has at most $L^{t+1}$ entries.
Since $|I|=O(n\cdot t)$, it suffices to bound the work per node.

For introduce and join nodes, each table entry is computed from a constant number of child entries, so the transition can be implemented in $O(1)$ time per state.
For forget nodes, although the recurrence takes a maximum over $L$ labels, the table can be computed in $O(L^{t+1})$ time by scanning all child states $\ell':X_j\to\lcal$ once, updating the parent key $\ell=\ell'_{X_i}$ via $\mathrm{DP}'_i(\ell)\gets \max\{\mathrm{DP}'_i(\ell),\mathrm{DP}_j(\ell')\}$, and finally setting $\mathrm{DP}_i(\ell)=\mathrm{DP}'_i(\ell)+\Phi_i(\ell)$.
Thus, every node can be processed in $O(L^{t+1})$ time, yielding a total DP running time of $O(|I|\cdot L^{t+1})=O(n\cdot t\cdot L^{t+1})$.

Computing all values $\Phi_i(\ell)$ over all nodes $i\in I$ and states $\ell:X_i\to\lcal$ requires $O(n\cdot L^{t+1})$ oracle evaluations, since the sets $V_i$ form a partition of $V$.
Under the representation assumptions of \Cref{sec: setup}, each query $f_v(\cdot)$ can be carried out in time $\mathrm{poly}(|\text{input}|)$.
Thus, the overall running time can be written as $O\bigl(n\cdot t\cdot L^{t+1}\bigr)\cdot \mathrm{poly}(|\text{input}|)$, and the preprocessing does not change the exponential dependence on $t$ and $L$.

Storing all tables uses $O(|I|\cdot L^{t+1})$ space.
If only the optimal value is required, the tables can be computed in a postorder traversal while discarding children after use, using $O(h\cdot L^{t+1})$ working space, where $h$ is the height of $T$.
Recovering an optimal labeling via backpointers requires $O(|I|\cdot L^{t+1})$ space in the worst case.

\subsection{Omitted proofs for inapproximability} \label{apd:inapprox-proofs}

\begin{proof}[Proof of Lemma~\ref{lem:gap-zero}]
Consider an instance produced by \Cref{lem:sat-star}, for which $F(x)\in\{0,1\}$ for every labeling $x$ and $\OPT\in\{0,1\}$.
Suppose $\mathcal{A}$ is a polynomial-time $\rho(n)$-approximation algorithm and let $\hat x$ be its output.
If the underlying formula is satisfiable, then $\OPT=1$ and thus $F(\hat x)\ge 1/\rho(n)>0$, which forces $F(\hat x)=1$.
If the formula is unsatisfiable, then $\OPT=0$ and hence $F(x)=0$ for all labelings $x$, implying $F(\hat x)=0$.

Therefore, given $\varphi$, we can construct the corresponding \name instance from \Cref{lem:sat-star}, run $\mathcal{A}$, and output \textsc{sat} iff $F(\hat x)=1$.
Since $F(\hat x)$ can be evaluated in polynomial time in the oracle model, this decides satisfiability in polynomial time, implying $\mathsf{P}=\mathsf{NP}$.
\end{proof}

\begin{proof}[Proof of Theorem~\ref{thm:inapprox}]
Suppose, for contradiction, that there exists a polynomial-time $\rho(n)$-approximation algorithm $\mathcal{A}$ for \name,
where $\rho(n)=n^{1-\varepsilon}$ for some $\varepsilon>0$ and $n\coloneqq |V|$.

Let $G=(V,E)$ be an arbitrary graph.
We construct an instance of \name on the same graph $G$ with binary label set $\lcal=\{1,2\}$.
For each $v\in V$, define
\[
f_v(x_{\Ncl{v}}) \coloneqq
\begin{cases}
1 & \text{if } x_v=2 \text{ and } x_u=1 \text{ for all } u\in \Ncl{v}\setminus\{v\},\\
0 & \text{otherwise.}
\end{cases}
\]
For a labeling $x\in\{1,2\}^V$, let $I(x) \coloneqq \{v\in V : x_v=2 \text{ and } x_u=1 \text{ for all } u\in \Ncl{v}\setminus\{v\}\}$.
Then $I(x)$ is an independent set in $G$ (since adjacent vertices cannot both satisfy the defining condition), and moreover
$F(x)=\sum_{v\in V} f_v(x_{\Ncl{v}})=|I(x)|$.
Hence $F(x)\le \alpha(G)$ for every labeling $x$.
Conversely, for any independent set $I\subseteq V$, the labeling $x=\mathbf{1}_I$ (setting $x_v=2$ iff $v\in I$) satisfies $F(x)=|I|$.
Therefore $\OPT=\alpha(G)$.

Apply $\mathcal{A}$ to this \name instance to obtain $\hat x$ with
$F(\hat x)\ge \OPT/\rho(n)=\alpha(G)/\rho(n)$.
Since $|I(\hat x)|=F(\hat x)$, the set $I(\hat x)$ is a $\rho(n)$-approximation for MIS on $G$.
This contradicts the $n^{1-\varepsilon}$-inapproximability of MIS (equivalently, \textsc{Max Clique} on the complement graph) unless $\mathsf{P}=\mathsf{NP}$ \cite{mis-hardness-zuckerman2006linear}.
\end{proof}

\subsection{Proof of Theorem~\ref{thm:color-approx}} \label{apd:color-proof}

\begin{proof}[Proof of Theorem~\ref{thm:color-approx}]
Let $x^\star$ be an optimal labeling.
For each color $c\in[q]$, define $\mathrm{OPT}_c \coloneqq \sum_{v\in C_c} f_v\bigl(x^\star_{\Ncl{v}}\bigr)$.
Since $\{C_c\}_{c\in[q]}$ partitions $V$, we have $\sum_{c=1}^q \mathrm{OPT}_c = \OPT$.
Hence there exists $c\in[q]$ with $\mathrm{OPT}_c\ge \OPT/q$.
Fix such a color $c$.
For every $v\in C_c$, we have $m_v \ge f_v\bigl(x^\star_{\Ncl{v}}\bigr)$ by definition of $m_v$ as a maximum over all assignments to $\Ncl{v}$.
Therefore,
\[
M_c \;=\; \sum_{v\in C_c} m_v \;\ge\; \sum_{v\in C_c} f_v\bigl(x^\star_{\Ncl{v}}\bigr) \;=\; \mathrm{OPT}_c \;\ge\; \OPT/q.
\]
By choice of $c^\star$, we obtain $M_{c^\star}\ge M_c\ge \OPT/q$.
By construction of $\hat x$ and \Cref{lem:disjoint-neighborhoods}, for every $v\in C_{c^\star}$ we have $\hat x_{\Ncl{v}}=a_v$ and thus $f_v\bigl(\hat x_{\Ncl{v}}\bigr)=f_v(a_v)=m_v$.
Since all local rewards are nonnegative, we conclude
\[
F(\hat x)\;=\;\sum_{v\in V} f_v\bigl(\hat x_{\Ncl{v}}\bigr)\;\ge\;\sum_{v\in C_{c^\star}} f_v\bigl(\hat x_{\Ncl{v}}\bigr)\;=\;\sum_{v\in C_{c^\star}} m_v\;=\;M_{c^\star}\;\ge\;\OPT/q.
\]
The running time follows by enumerating all $L^{|\Ncl{v}|}$ assignments to each $\Ncl{v}$ to compute $m_v$ and a maximizer $a_v$, plus $O(n)$ overhead to compute the scores $M_c$ and assemble $\hat x$.
\end{proof}

\begin{corollary}[Greedy coloring instantiation and bounded-degree regime]
\label{cor:color-approx-greedy}
Without a coloring as input, one can compute a greedy coloring of $G^2$ using at most $\Delta(G^2)+1$ colors \cite{west2001introduction}.
Since every vertex has at most $\Delta(G)^2$ vertices within distance at most $2$ in $G$, we have $\Delta(G^2)\le \Delta(G)^2$ and obtain a $(\Delta(G)^2+1)$-approximation.
The coloring overhead is polynomial in $n$.
\end{corollary}

\subsection{Running time details for Theorem~\ref{thm:planar-ptas}} \label{apd:ptas-runtime}

We justify the running-time bound in the proof of \Cref{thm:planar-ptas}.

Fix an offset $s$ and a connected component $C$ of $G_s$.
Contract all vertices in BFS layers closer to $r_0$ than $C$ into a single vertex $z$ (and delete all other components) to obtain a planar graph $\widetilde{G}$ containing $C\cup\{z\}$.
For every $v\in C$, following the BFS tree from $v$ toward $r_0$ reaches a vertex in $R_s$ after at most $k-1$ steps, hence $\dist_{\widetilde{G}}(v,z)\le k-1$.
Thus the (graph-theoretic) radius of $\widetilde{G}[C\cup\{z\}]$ is at most $k-1$.
Planar graphs have linear local treewidth, and in particular treewidth bounded by a function of the radius \cite{eppstein2000diameter}, hence $\tw(C)=O(k)$.

By the bag-expansion bound discussed in \Cref{subsec:twG2}, we have
\[
\tw(C^2)\;\le\;(\tw(C)+1)\bigl(\Delta(C^2)+1\bigr)-1
\;\le\;(\tw(C)+1)\bigl(\Delta^2+1\bigr)-1
\;=\;O(k\Delta^2),
\]
We use $\Delta(C^2)\le \Delta(C)^2\le \Delta^2$.
Applying CFDP (\Cref{thm:cfdp-correct}) to each component costs $O\!\left(|C|\cdot k\Delta^2\cdot L^{O(k\Delta^2)}\right)\cdot n^{O(1)}$ time.
Summing over components gives $O\!\left(n\cdot k\Delta^2\cdot L^{O(k\Delta^2)}\right)\cdot n^{O(1)}$ per offset.
Multiplying by the $k$ offsets yields the stated running time.

\begin{remark}[Beyond radius-$1$ neighborhoods]
The constant $3$ in the shifting loss arises from radius-$1$ dependence: a term $f_v$ is affected only if $v$ lies within graph distance $1$ of the removed layer.
More generally, if each $f_v$ depends on all vertices within distance at most $p$ of $v$ in $G$ (equivalently, on $G^p$ as in \Cref{sec: setup}), the same argument yields a $(1-(2p+1)/k)$-approximation by discarding the $(2p+1)$ layers within distance $p$ of the removed class.
\end{remark}

\section{Budgeted monotone submodular variant} \label{apd:submod}

Many interference and activation/coverage models use binary decisions (treat/activate) subject to a cardinality budget, where each node's utility depends on its own status and on an \emph{exposure} induced by treated neighbors (e.g., count- or threshold-based exposure mappings) \cite{aronowSamii2017interference,kempe2003maximizing,krause2014submodular}.
In such regimes, diminishing-returns structure is often natural; below we give a \emph{local} sufficient condition that implies global submodularity of the NAGL objective.

We restrict to binary labels $\lcal=\{0,1\}$ and identify a labeling $x\in\{0,1\}^V$ with the active set $S\coloneqq\{v\in V:x_v=1\}$.
For each $v\in V$, define the induced local set function $g_v:2^{\Ncl{v}}\to\mathbb{R}_{\ge 0}$ (where $\Ncl{v}$ is the closed neighborhood in $G$) by
$g_v(T)\coloneqq f_v(\mathbf{1}_T)$, where $\mathbf{1}_T$ is the $\{0,1\}$-assignment on $\Ncl{v}$ that is $1$ on $T$ and $0$ otherwise.
Then the NAGL objective can be written as $F(S)\;=\;\sum_{v\in V} g_v(S\cap \Ncl{v})$.
Given a budget $K$, our goal is to solve $\max\{F(S): S\subseteq V,\ |S|\le K\}$.

We assume that each $g_v$ is monotone and submodular on ground set $\Ncl{v}$.

The next lemma lifts local diminishing returns on each $\Ncl{v}$ to global submodularity of $F$.

\begin{lemma}[Local-to-global submodularity]
\label{lem:local-to-global-submod}
If every $g_v$ is monotone and submodular, then $F:2^V\to\mathbb{R}_{\ge 0}$ is monotone and submodular.
\end{lemma}
\begin{proof}
Fix $A\subseteq B\subseteq V$ and $s\in V\setminus B$.
For each $v$, let $A_v\coloneqq A\cap \Ncl{v}$ and $B_v\coloneqq B\cap \Ncl{v}$, so $A_v\subseteq B_v$.
If $s\notin \Ncl{v}$ then the marginal contribution of $g_v$ is $0$ under both $A$ and $B$; otherwise, submodularity of $g_v$ gives
$g_v(A_v\cup\{s\})-g_v(A_v)\ge g_v(B_v\cup\{s\})-g_v(B_v)$.
Summing over $v$ yields diminishing returns for $F$; monotonicity follows similarly.
\end{proof}

We maximize $F(S)$ subject to $|S|\le K$ via Standard Greedy, which iteratively adds a vertex with maximum marginal gain $\mathrm{marg}(u\mid S)\coloneqq F(S\cup\{u\})-F(S)$.
Only local terms for vertices $w$ with $u\in \Ncl{w}$ can change when adding $u$, and these vertices are exactly $w\in\Ncl{u}$.
Thus $\mathrm{marg}(u\mid S)$ can be computed using only oracle evaluations on the neighborhoods $\{\Ncl{w}: w\in\Ncl{u}\}$.

\begin{theorem}[Nemhauser--Wolsey--Fisher \cite{nemhauser_1978}]
Let $S_{\mathrm{gr}}$ be the set returned by Standard Greedy.
Under the assumptions of Lemma~\ref{lem:local-to-global-submod}, $F(S_{\mathrm{gr}})\;\ge\;\left(1-\frac{1}{e}\right)\cdot \max\{F(S): S\subseteq V,\ |S|\le K\}$.
\end{theorem}

\section{Experiments} \label{apd:experiments}

We evaluate CFDP and an ILP baseline in three scenarios.
In the budgeted monotone submodular regime, we additionally evaluate Standard Greedy from \Cref{apd:submod}.
Scenarios~1--2 compare \emph{exact} solvers (CFDP vs.\ ILP) on (i) road-network subgraphs, where $\tw(G^2)$ grows with $n$, and (ii) $2\times r$ ladder graphs, where $\tw(G^2)$ is constant.
Scenario~3 studies the \emph{budgeted monotone submodular} variant on large road subgraphs, comparing Standard Greedy with ILP.

\subsection{Experimental setup and methods}
CFDP is our exact DP on a tree decomposition of $G^2$; runtimes are end-to-end, including constructing $G^2$ and computing the decomposition.
We compute tree decompositions using an implementation of the algorithm of \cite{tamaki2017positiveinstance}.
As an exact baseline we formulate \name as an ILP (described below), solved with Gurobi Optimizer~11.0.0.
We impose a fixed wall-clock time limit of 10 hours per run; timed-out instances are omitted from the plots.

\paragraph*{Road datasets and BFS-induced subgraphs.}
Scenarios~1 and~3 use road networks from the Network Data Repository (NRVis) \cite{rossi2015networkrepository}.
We use the MatrixMarket datasets \texttt{road-minnesota}\footnote{\url{https://nrvis.com/./download/data/road/road-minnesota.zip}} and \texttt{road-usroads}\footnote{\url{https://nrvis.com/./download/data/road/road-usroads.zip}}.
We treat these datasets as undirected, unweighted graphs.
Given a starting vertex $s$, we compute a breadth-first search (BFS) order with deterministic tie-breaking by sorting neighbors by their identifiers.
We fix $s$ for each dataset and reuse the resulting BFS order across all values of $n$.
We use $s=2417$ for \texttt{road-minnesota} and $s=58610$ for \texttt{road-usroads}.
For a target size $n$, we take the first $n$ vertices in this order and consider the induced subgraph.
This produces connected subgraphs and ensures that increasing $n$ yields nested instances.

\paragraph*{Scenario 1 (Minnesota road; synthetic rewards).}
We use BFS-induced subgraphs of the Minnesota road network (\texttt{road-minnesota}) with label set size $L=3$.
We vary $n$ between 40 and 500.
Local reward functions are synthetic random tables over $\lcal^{|\Ncl{v}|}$.

\paragraph*{Scenario 2 (ladder graphs; bounded $\tw(G^2)$).}
We use $2\times r$ ladder graphs with $n=2r$ and label set size $L=3$.
We vary $r$ between 20 and 10000 (i.e., $n$ between 40 and 20000).
Local reward functions are synthetic random tables over $\lcal^{|\Ncl{v}|}$.

\paragraph*{Scenario 3 (US roads; budgeted monotone submodular).}
We use BFS-induced subgraphs of the US road network (\texttt{road-usroads}) with binary labels $L=2$.
We vary $n$ between 200 and 126146.
We impose a cardinality budget $K=\min\{\lfloor 0.02n\rfloor,400\}$.
For each $n$ we generate five random objectives and run five repetitions per objective.
The objective is constructed using a max-type local gadget with weights $w_{\mathrm{hi}}=100$ and $w_{\mathrm{lo}}$ sampled uniformly from $[30,90]$.
Concretely, for each vertex $v$ we sample a weight $w_{\mathrm{lo}}(v)\in[30,90]$ and define $g_v:2^{\Ncl{v}}\to\mathbb{R}_{\ge 0}$ by $g_v(T)\coloneqq 0$ for $T=\emptyset$ and, for $T\neq\emptyset$, $g_v(T)\coloneqq w_{\mathrm{hi}}$ if $v\in T$ and $g_v(T)\coloneqq w_{\mathrm{lo}}(v)$ otherwise.
We set $f_v(x_{\Ncl{v}})\coloneqq g_v(\{u\in\Ncl{v}: x_u=1\})$.
Since $g_v$ is a nonnegative combination of the submodular function $T\mapsto \mathbf{1}[T\neq\emptyset]$ and the modular function $T\mapsto \mathbf{1}[v\in T]$, it is monotone and submodular, hence the assumptions of \Cref{lem:local-to-global-submod} hold.

\paragraph*{ILP baseline.}
We use the following standard local-configuration ILP for \name.
For each vertex $v\in V$ and label $\ell\in\lcal$ we introduce a binary variable $x_{v,\ell}$ indicating $x_v=\ell$.
For each $v\in V$ and neighborhood assignment $a:\Ncl{v}\to\lcal$ we introduce a binary variable $y_{v,a}$ selecting this local configuration.
\begin{equation}
\label{eq:ilp-local-config}
\begin{aligned}
\max\ & \sum_{v\in V}\ \sum_{a:\Ncl{v}\to\lcal} f_v(a)\,y_{v,a} \\
\text{s.t.}\ & \sum_{\ell\in\lcal} x_{v,\ell}=1 && \forall v\in V, \\
& \sum_{a:\Ncl{v}\to\lcal} y_{v,a}=1 && \forall v\in V, \\
& \sum_{a:\Ncl{v}\to\lcal,\,a(u)=\ell} y_{v,a}=x_{u,\ell} && \forall v\in V,\ u\in\Ncl{v},\ \ell\in\lcal, \\
& \sum_{v\in V} x_{v,1}\le K && \text{(Scenario~3 only)}, \\
& x_{v,\ell}\in\{0,1\},\ y_{v,a}\in\{0,1\} && \forall v,\ell,a.
\end{aligned}
\end{equation}

\subsection{Results}
\Cref{fig:exp12-runtime} compares CFDP with the ILP baseline.
In Scenario~1 (left), we use Minnesota road subgraphs with random rewards ($L=3$), where $\tw(G^2)$ grows with $n$.
CFDP is consistently faster than ILP on all sizes solvable by both methods (e.g., $15$ seconds vs.\ $16$ minutes at $n=200$), before both eventually time out (CFDP at $n=340$ and ILP at $n=260$).
In Scenario~2 (right), we use $2\times r$ ladder graphs with constant $\tw(G^2)$ (here $\tw(G^2)=4$ for all $r$).
CFDP scales smoothly up to $n=20000$ (about $20$ seconds), whereas ILP already hits the time limit at $n=120$, matching the predicted dependence on $\tw(G^2)$.

In Scenario~3 (\Cref{fig:exp3-runtime,fig:exp3-ratio}), we consider the budgeted monotone submodular variant ($L=2$) on subgraphs of US road network with budget $K=\min\{\lfloor 0.02n\rfloor,400\}$.
Greedy remains below about $21$ seconds even at $n=126{,}146$, while ILP takes about $30$ minutes (\Cref{fig:exp3-runtime}).
Greedy achieves approximation ratios around $0.8$ with little variation across $n$, exceeding the worst-case guarantee $1-1/e$ (\Cref{fig:exp3-ratio}).
We omit CFDP in Scenario~3 because it targets the unbudgeted additive/local objective.
Extending it to a budgeted monotone submodular objective would both introduce a factor $K$ in the DP state space and require redesigning the DP to handle a non-local set function $F$, whereas the ILP adds the budget via one constraint and Greedy is the standard approach.

\begin{figure}[tb]
    \centering
    \begin{subfigure}[t]{0.5\linewidth}
        \centering
        \includegraphics{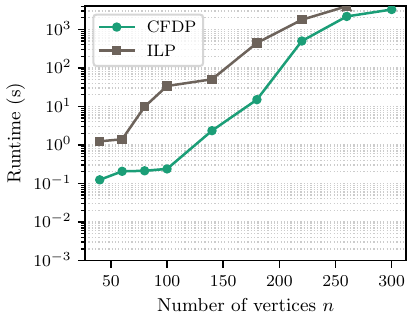}
        \caption{Minnesota road subgraphs.}
        \label{fig:exp1-runtime}
    \end{subfigure}%
    \begin{subfigure}[t]{0.5\linewidth}
        \centering
        \includegraphics{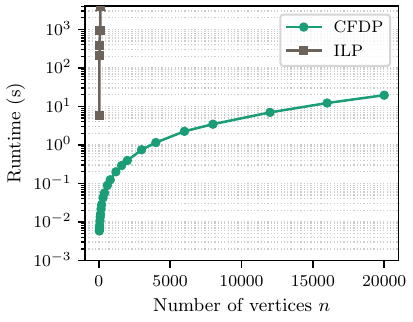}
        \caption{Ladder graphs. Here $\tw(G^2)=4$ for all $r$.}
        \label{fig:exp2-runtime}
    \end{subfigure}
    \caption{Scenario~1--2 end-to-end runtime comparison of CFDP and the ILP baseline.
    End-to-end includes constructing $G^2$ and computing the decomposition.
    Each point is the mean over five repetitions, with shaded bands showing one standard deviation.
    }
    \label{fig:exp12-runtime}
\end{figure}

\begin{figure}[tb]
    \centering
    \begin{subfigure}[t]{0.5\linewidth}
        \centering
        \includegraphics{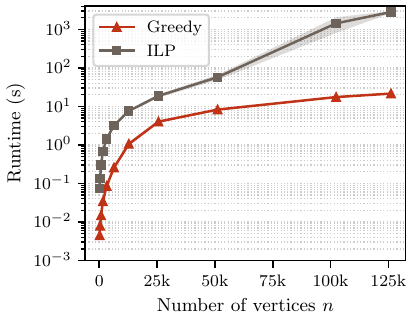}
        \caption{Runtime.}
        \label{fig:exp3-runtime}
    \end{subfigure}%
    \begin{subfigure}[t]{0.5\linewidth}
        \centering
        \includegraphics{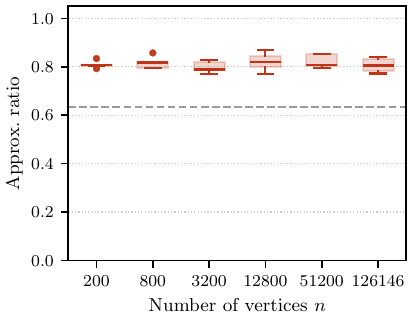}
        \caption{Approx.\ ratio.}
        \label{fig:exp3-ratio}
    \end{subfigure}
    \caption{Scenario~3 results in the budgeted monotone submodular setting.
    Left: runtime comparison of Greedy and ILP (mean $\pm$ one s.d.\ over $25$ runs).
    Right: box plot of the approximation ratio $F(S_{\mathrm{gr}}) / \OPT$ over the same runs.
    The dashed line marks $1-1/e$.
    }
    \label{fig:exp3}
\end{figure}

\paragraph*{Representative graph visualizations.}
\Cref{fig:app-exp1-bfs} and \Cref{fig:app-exp3-bfs} visualize BFS-induced subgraphs from Scenarios~1 and~3.
The BFS start vertex $s$ is highlighted.
For Scenario~3 we visualize only the first 2000 BFS vertices for readability.

\begin{figure}[tb]
    \centering
    \includegraphics{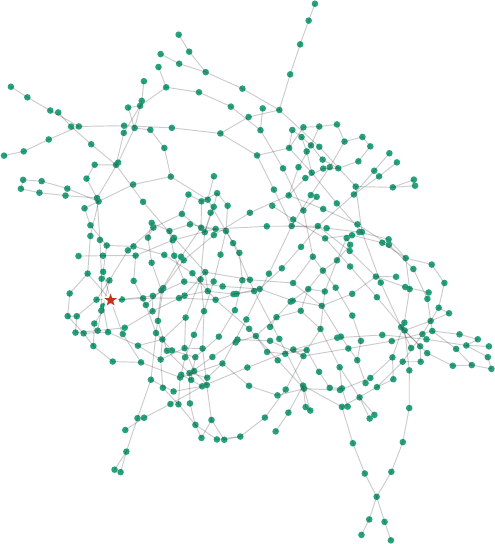}
    \caption{Scenario~1: BFS-induced subgraph of the Minnesota road network (\texttt{road-minnesota}) at $n=350$ and $m=413$.
    The red star marks the BFS start vertex.
    }
    \label{fig:app-exp1-bfs}
\end{figure}

\begin{figure}[tb]
    \centering
    \includegraphics{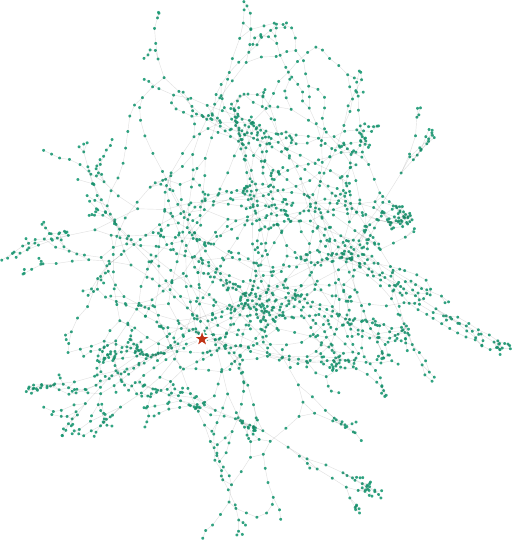}
    \caption{Scenario~3: BFS-induced subgraph of the US road network (\texttt{road-usroads}) at $n=120000$ and $m=153{,}781$.
    The red star marks the BFS start vertex.
    For readability we visualize only the first 2000 BFS vertices.
    }
    \label{fig:app-exp3-bfs}
\end{figure}

\end{document}